\documentclass{article}[11pt]
\usepackage{geometry}
\usepackage{amssymb}
\usepackage{amsthm}
\usepackage{amsmath}
\usepackage{amsfonts}
\usepackage{bbm}
\usepackage{xcolor}
\usepackage{hyperref}

\usepackage[hyphenbreaks]{breakurl}
\usepackage{microtype}
\usepackage{breakcites}
\usepackage{comment}
\definecolor{ceruleanblue}{rgb}{0.16, 0.32, 0.75}
\definecolor{darkmidnightblue}{rgb}{0.0, 0.2, 0.4}
\definecolor{darkpastelgreen}{rgb}{0.01, 0.75, 0.24}
\definecolor{bleudefrance}{rgb}{0.19, 0.55, 0.91}
\hypersetup{
	colorlinks   = true,
	citecolor    = bleudefrance,
	linkcolor	  = darkpastelgreen,
	urlcolor     = bleudefrance
}
\usepackage[capitalize, nameinlink]{cleveref}
\usepackage{graphicx}
\usepackage{subcaption}
\usepackage{verbatim}
\usepackage{bm}
\usepackage{enumitem}
\usepackage{tabularx}

\makeatletter 
\let\c@table\c@figure
\let\c@lstlisting\c@figure
\makeatother

\RequirePackage[T1]{fontenc} \RequirePackage[tt=false, type1=true]{libertine} \RequirePackage[varqu]{zi4} \RequirePackage[libertine]{newtxmath}

\newcommand{\myparagraph}[1]{\noindent{\textbf{#1.}}}
\newcommand{\ind}[1]{\mathbbm{1}[#1]}
\newcommand{\bs}[1]{\boldsymbol{#1}}
\newcommand{\bv}[1]{\mathbf{#1}}
\newcommand{\Var}{\operatorname{Var}}

\DeclareMathOperator*{\E}{\mathbb{E}}
\DeclareMathOperator*{\R}{\mathbb{R}}
\DeclareMathOperator*{\argmin}{arg\,min}
\DeclareMathOperator*{\argmax}{arg\,max}
\DeclareMathOperator*{\sign}{sign}
\let\norm\relax
\newcommand{\norm}[1]{\|#1\|}

\makeatletter
\newtheorem*{rep@theorem}{\rep@title}
\newcommand{\newreptheorem}[2]{%
	\newenvironment{rep#1}[1]{%
		\def\rep@title{\Cref{##1}}%
		\begin{rep@theorem}}%
		{\end{rep@theorem}}}
\makeatother

\newtheorem{theorem}{Theorem}
\newreptheorem{theorem}{Theorem}
\newtheorem{lemma}{Lemma}
\newreptheorem{lemma}{Lemma}

\newtheorem{fact}[theorem]{Fact}

\usepackage{algorithmicx}
\usepackage{algorithm}
\usepackage{algpseudocode}
\usepackage{flushend}
\usepackage{booktabs}

\makeatletter
\newcommand{\algrule}[1][.2pt]{\par\vskip.2\baselineskip\hrule height #1\par\vskip.2\baselineskip}
\makeatother

  \usepackage{nth}
  \usepackage{intcalc}

  \newcommand{\cAAAI}[1]{AAAI\ Conference\ on\ Artificial (AAAI)}

\title{\vspace{-.5em}Weighted Minwise Hashing Beats Linear Sketching for Inner Product Estimation
	}
\author{
    	Aline Bessa$^1$\footnote{Author names are listed in alphabetical order.}\\ {aline.bessa@nyu.edu}
    	\and 
    	Majid Daliri$^1$\\ {daliri.majid@nyu.edu}
    	\and
    	Juliana Freire$^1$\\ {juliana.freire@nyu.edu}
    	\and
    	Cameron Musco$^2$\\ {cmusco@cs.umass.edu}
    	\and
    	Christopher Musco$^1$\\ {cmusco@nyu.edu}
    	 \and
    	A\'{e}cio Santos$^1$\\ {aecio.santos@nyu.edu}
    	\and
    	Haoxiang Zhang$^1$\\ {haoxiang.zhang@nyu.edu}
}
\date{%
	$^1$New York University\\%
	$^2$University of Massachusetts Amherst\\[2ex]%
	\today
}

\begin{document}
\maketitle

\begin{abstract}
    We present a new approach for independently computing compact sketches that can be used to approximate the {inner product} between pairs of high-dimensional vectors. Based on the Weighted MinHash algorithm, our approach admits strong accuracy guarantees that improve on the guarantees of popular \emph{linear sketching} approaches for inner product estimation, such as CountSketch and Johnson-Lindenstrauss projection. Specifically, while our method exactly matches linear sketching for dense vectors, it yields significantly \emph{lower} error  for sparse vectors with limited overlap between non-zero entries. Such vectors arise in many applications involving sparse data, as well as in increasingly popular dataset search applications, where inner products are used to estimate data covariance, conditional means, and other quantities involving columns in \emph{unjoined tables}. We complement our theoretical results by showing that our approach empirically outperforms existing linear sketches and unweighted hashing-based sketches for sparse vectors.
\end{abstract}

\section{Introduction}
The \textit{inner product} of two vectors $\bv{a}$ and $\bv{b}$,
$\langle \bv{a}, \bv{b} \rangle = \sum_{k=1}^n \bv{a}[k]\bv{b}[k]$, is a ubiquitous operation. Among many other applications, inner products can be used to compute document similarities~\cite{SaltonWongYang:1975}, to evaluate learned classification models, and to estimate join sizes
~\cite{AlonMatiasSzegedy:1999,Achlioptas:2003,RusuDobra:2008}.
However, in modern applications involving very high-dimensional vectors, computing exact inner products can be intractable. The computational cost is $O(n)$ and computing $\langle \bv{a}, \bv{b}\rangle$ requires loading $O(n)$ numbers from memory, or communicating $O(n)$ numbers if $\bv{a}$ and $\bv{b}$ are stored on different machines.

A common approach for resolving this issue is to pre-compute a small space compression (a sketch) of each vector, which we will denote by $\mathcal{S}(\bv{a})$ and $\mathcal{S}(\bv{b})$, respectively. An estimation function $\mathcal{F}$ is then used to approximate the inner product as $\mathcal{F}\left(\mathcal{S}(\bv{a}), \mathcal{S}(\bv{b})\right) \approx \langle \bv{a}, \bv{b} \rangle$.
 The beauty of sketching is that it simultaneously reduces storage, communication, and runtime complexity. Moreover, once computed, sketches can be reused again and again to estimate inner products with other vectors. For example, given another vector $\bv{c}$ we can estimate $\langle \bv{a}, \bv{c} \rangle \approx \mathcal{F}\left(\mathcal{S}(\bv{a}), \mathcal{S}(\bv{c})\right)$. 

 Sketching methods for approximating inner products are already widely used throughout computer science. In machine learning, they can be used to accelerate the training of large-scale linear models like support vector machines or logistic regression 
\cite{ArriagaVempala:2006,LiShrivastavaMoore:2011}. In relational databases, inner product sketches are used in query optimizers to choose optimal query plans without having to execute expensive queries that involve large joins \cite{CormodeGarofalakisHaas:2011}. More recently, inner product sketches have found applications in dataset search and discovery, where they are used to discover joinable tables \cite{FernandezMinNava:2019} and to estimate other column statistics, such as correlation~\cite{SantosBessaChirigati:2021}, without explicitly performing a {join} operation between two tables. We discuss these applications and others in \Cref{sec:apps}.

\smallskip
\noindent\textbf{What was Previously Known?}
In all of the applications above, a primary concern is optimizing the trade-off between the sketch size (which governs storage, communication, and runtime efficiency) and how accurately $\mathcal{F}\left(\mathcal{S}(\bv{a}), \mathcal{S}(\bv{b})\right)$ approximates 
$\langle \bv{a}, \bv{b} \rangle$. A large sketch size will in general lead to better approximation, but the question is by exactly how much. 
Currently, the only methods with strong theoretical guarantees on this tradeoff 
for \textit{general vectors} (i.e., vectors without any assumed value distribution or magnitude) are based on \emph{linear sketching} algorithms. Such algorithms include the famous ``tug-of-war'' sketch, a.k.a. the AMS sketch \cite{AlonMatiasSzegedy:1999,AlonGibbonsMatias:1999}, the CountSketch algorithm \cite{CharikarChenFarach-Colton:2002}, and methods based on Johnson-Lindenstrauss (JL) random projection~\cite{Achlioptas:2003,DasguptaGupta:2003}.

All of these approaches have a similar form. We choose a random matrix $\bv{\Pi} \in \R^{m\times n}$ ($\bv{\Pi}$ might have i.i.d. random entries or more complex structure) and sets $\mathcal{S}(\bv{a})= \bv{\Pi}\bv{a}$ and $\mathcal{S}(\bv{b}) = \bv{\Pi}\bv{b}$.
Each sketch is a length $m$ vector and is considered a \emph{linear sketch} since $\mathcal{S}$ is a linear function.
To estimate the inner product, the typical approach is to simply return the sketch inner product $\langle \mathcal{S}(\bv{a}), \mathcal{S}(\bv{b})\rangle$.\footnote{Other estimators involving e.g., the median of multiple approximate inner products, are also used \cite{LarsenPaghTetek:2021}. However, theoretical guarantees are similar, typically differing in the dependence on the failure probability $\delta$}

\begin{table*}[t]
\vspace{-.75em}
	\def\arraystretch{1}
	\centering
	\begin{tabular}{ p{0.37\linewidth}|p{0.35\linewidth}|p{0.16\linewidth}}
		\toprule
		\textbf{Method}  & \textbf{Error for sketches of size $O(1/\epsilon^2)$} & \textbf{Assumptions}\\
		\midrule
		Johnson-Lindenstrauss Projection \cite{ArriagaVempala:2006}, AMS \cite{AlonMatiasSzegedy:1999}, CountSketch \cite{CharikarChenFarach-Colton:2002} & $\epsilon \cdot\|\bv{a}\| \|\bv{b}\|$ & None \\ 
				\midrule
		MinHash (MH) Sampling \cite{BeyerHaasReinwald:2007} & $\epsilon \cdot \max\left(\|\bv{a}_{\mathcal{I}}\|\|\bv{b}\|, \|\bv{a}\|\|\bv{b}_{\mathcal{I}}\| \right)$ & $\bv{a}, \bv{b}$ are binary, i.e. with $\{0,1\}$ entries. \\
		\midrule
		\textbf{Our Method: Weighted MinHash (WMH) Sampling} & $\epsilon \cdot \max\left(\|\bv{a}_{\mathcal{I}}\|\|\bv{b}\|, \|\bv{a}\|\|\bv{b}_{\mathcal{I}}\| \right)$ & None\\
		\bottomrule
	\end{tabular}
    \vspace{.25em}
	\caption{Comparison of high-probability additive error guarantees for estimating $\langle \bv{a}, \bv{b}\rangle$ using various sketching methods.
	We let  $\mathcal{I} = \{i: \bv{a}[i] \neq 0\text{ and } \bv{b}[i] \neq 0\}$ denote the intersection of $\bv a$'s  and $\bv{b}$'s supports. $\bv{a}_{\mathcal{I}}$ and $\bv{b}_{\mathcal{I}}$ are $\bv{a}$ and $\bv{b}$ restricted to indices in $\mathcal{I}$.
	Since $\max\left(\|\bv{a}_{\mathcal{I}}\|\|\bv{b}\|, \|\bv{a}\|\|\bv{b}_{\mathcal{I}}\| \right) \leq \|\bv{a}\| \|\bv{b}\|$, the bound for our Weighted MinHash (WMH) method always beats the linear sketching methods. Our bound matches that of unweighted MinHash, but without the strong limiting assumption that $\bv{a}$ and $\bv{b}$ are binary; it holds for all vectors.}
	\label{tab:error_guarantees}
\end{table*}

 A textbook theoretical accuracy guarantee for inner product estimation based on linear sketching is:
\begin{fact}[Linear Sketching for Inner Products \cite{ArriagaVempala:2006}]
	\label{fact:jl_result}
	Let $\epsilon,\delta \in (0,1)$ be accuracy and failure probability parameters respectively and let $m = O(\log(1/\delta)/\epsilon^2)$.
	Let $\bv{\Pi} \in \R^{m \times n}$ be a random matrix with each entry set independently to $+\sqrt{1/m}$ or $-\sqrt{1/m}$ with equal probability. For length $n$ vectors $\bv{a},\bv{b} \in \R^n$, let  $\mathcal{S}(\bv{a}) = \bv{\Pi}\bv{a}$ and $\mathcal{S}(\bv{b}) = \bv{\Pi}\bv{b}$. With probability at least $1-\delta$,
    \vspace{-.2em}
	\begin{align*}
		\left |  \langle \mathcal{S}(\bv{a}), \mathcal{S}(\bv{b})\rangle - \langle \bv{a}, \bv{b} \rangle \right | \leq \epsilon \|\bv{a}\| \|\bv{b}\|
	\end{align*}
    \vspace{-.25em}
	where $\norm{\bv x}$ denotes the standard Euclidean norm.
\end{fact}
In addition to dense random matrices, analogous results to \Cref{fact:jl_result} can be proven for sparse JL matrices, CountSketch matrices, and other linear sketches \cite{CormodeGarofalakisHaas:2011}. 
The fact provides a powerful accuracy guarantee that improves with the sketch size $m$ and depends naturally on the norms of $\bv{a}$ and $\bv{b}$. To the best of our knowledge, linear sketching methods were previously the only known algorithms to obtain such a strong theoretical guarantee.

\subsection{Our Contributions}
In this paper we introduce a novel method for inner product sketching based on the Weighted MinHash sketch \cite{GollapudiPanigrahy:2006,ManasseMcSherryTalwar:2010,Ioffe:2010}, which is a variant of the classic MinHash method \cite{Broder:1997,BroderCharikarFrieze:1998}. We prove that our method obtains a refined guarantee than \Cref{fact:jl_result}. In particular, it matches the result for linear sketches in the worst case when $\bv{a}$ and $\bv{b}$ are dense\footnote{For dense vectors, \Cref{fact:jl_result} is actually optimal up to constants: recent work implies that no sketch of size $m = o(\log(1/\delta)/\epsilon^2)$ can achieve error $\epsilon \|\bv{a}\|\|\bv{b}\|$ with probability $1-\delta$ for all inputs  \cite{LarsenNelson:2017,AlonKlartag:2017}. Our result also matches this lower bound.}, but always obtains a \emph{better bound} when  $\bv{a}$ and $\bv{b}$ are sparse vectors with limited overlap between non-zero entries.  As discussed further in \Cref{sec:apps}, such pairs of vectors are the norm in many applications of inner product sketching to database problems and modern dataset search applications.
\begin{theorem}[Main Result]
	\label{thm:main}
		Let $\epsilon,\delta \in (0,1)$ be accuracy and failure probability parameters and let $m = O(\log(1/\delta)/\epsilon^2)$. There is an algorithm $\mathcal{S}$ 
that produces size-$m$ sketches (\Cref{alg:weighted_sketch}), 		
		along with an estimation procedure $\mathcal{F}$ (\Cref{alg:weight_est}), such that for any $\bv a,\bv b \in \R^n$, with probability at least $1-\delta$,
	\begin{align*}
		\left | \mathcal{F}\left(\mathcal{S}(\bv{a}), \mathcal{S}(\bv{b})\right) - \langle \bv{a}, \bv{b} \rangle\right | \leq  \epsilon \max\left(\|\bv{a}_{\mathcal{I}}\|\|\bv{b}\|, \|\bv{a}\|\|\bv{b}_{\mathcal{I}}\| \right)
	\end{align*}
	Above, $\mathcal{I} = \{i: \bv{a}[i] \neq 0\text{ and } \bv{b}[i] \neq 0\}$ is the intersection of $\bv a$'s  and $\bv{b}$'s supports. $\bv{a}_{\mathcal{I}}$ and $\bv{b}_{\mathcal{I}}$ denote $\bv{a}$ and $\bv{b}$ restricted to indices in $\mathcal{I}$. 
\end{theorem}
We  always have $\|\bv{a}_{\mathcal{I}}\| \leq \|\bv{a}\|$ and  $\|\bv{b}_{\mathcal{I}}\| \leq \|\bv{b}\|$, so we can bound $ \max\left(\|\bv{a}_{\mathcal{I}}\|\|\bv{b}\|, \|\bv{a}\|\|\bv{b}_{\mathcal{I}}\| \right) \leq \|\bv{a}\|\|\bv{b}\|$.
That is, the guarantee of \Cref{thm:main} matches that of  \Cref{fact:jl_result} in the worse-case, but can be significantly better. 
For example, consider $\bv{a}$ and $\bv{b}$ that have roughly the same number of non-zero entries, but only a $\gamma < 1$ fraction of those entries are non-zero in \emph{both} $\bv{a}$ and $\bv{b}$. In this case, it is reasonable to expect that 
$\|\bv{a}_{\mathcal{I}}\|^2 \approx \gamma \|\bv{a}\|^2$ and $\|\bv{b}_{\mathcal{I}}\|^2 \approx \gamma\| \bv{b}\|^2$ since $\bv{a}_{\mathcal{I}}$ and $\bv{b}_{\mathcal{I}}$ contain just a $\gamma$ fraction of entries from the original vectors. Our course, the actually improvement is data dependent; for example, we might have that $\|\bv{a}_{\mathcal{I}}\|^2$ is significantly smaller than $\gamma \|\bv{a}\|^2$, or that it is not much smaller than $\|\bv{a}\|^2$. 

Nevertheless, considering the ``typical case'' when a $\gamma$ fraction of non-zeros overlap, we might expect the bound from \Cref{thm:main} to be better than \Cref{fact:jl_result} by a factor of $\sqrt{\gamma}$. 
So, to obtain the same error as a linear sketch, our method could set
$m$ smaller by a factor of $\gamma$. 
In many applications, $\gamma$ is very small. E.g., in \Cref{sec:exp} we consider a document similarity problem where $\gamma \leq .05$ for $95\%$ of vector pairs sketched. This could equate to roughly a
  $20x$ improvement in sketch size required to achieve a specified level of error. 

Thanks to their strong theoretical guarantees, linear sketching algorithms have become the go-to approach for generic inner product estimation \cite{CormodeGarofalakisHaas:2011}. Our results show for the first time that an alternative method can provide stronger bounds. We hope that this paper will serve as a starting point for further investigation into hashing-based algorithms for inner product sketching. 

\vspace{-.25em}
\subsection{Motivating Application: Dataset Search} 
\label{sec:apps}
Before presenting the technical details of our results and discussing related work, we detail one application that could benefit from our proposed sketches, and helps illustrate the importance of obtaining bounds for inner product estimation that are sensitive to the number of overlapping non-zero entries in $\bv{a}$ and $\bv{b}$. 
Specifically, we consider the problem of \emph{dataset search} which has received increasing attention in recent years~\cite{LehmbergEtAl:2014,ZhuDengNargesianMiller:2019,YangZhangZhang:2019,ZhuNargesianPu:2016,FernandezMinNava:2019,SantosBessaChirigati:2021,SantosBessaMusco:2022}. 

Suppose that a data scientist wants to understand the reasons for fluctuations in taxi ridership in New York City in 2022. The analyst only has a table containing two columns: a \textit{date} column and the \emph{number of taxi rides} taken on that day. In order to carry out the analysis, she needs to find other tables, either in her organization's data lake or in public repositories like NYC Open Data  (which contain thousands of datasets \cite{New-York:2022}), that would bring in other relevant variables when 
joined with the original~table. For example, the analyst might hope to find weather data, which can impact taxi ridership.
Moreover, she would like to find relevant factors that she \emph{might not think of on her own}, in an \emph{automatic way}.

To solve this problem, we would like methods to automatically discover tables that are both 1) joinable with the target table (i.e., also contain columns with dates from 2022) and 2) meaningfully related with the analyst's data. For example, a table containing precipitation data should be returned if taxi ridership is significantly higher or lower on days with high precipitation. To find such tables, brute force search is not infeasible -- we typically cannot afford to join the analyst's table with all tables in the search set to look for good candidates. Instead, we need to efficiently estimate statistics between disparate tables \emph{without materializing their join} \cite{SantosBessaChirigati:2021}.

\begin{figure}[t]
\centering
\footnotesize

\parbox{.225\linewidth}{
  \centering
  \begin{tabular}{cc}
    \multicolumn{2}{c}{$\mathcal{T}_A$}        \\
    \hline
    \textbf{$K_A$} & \textbf{$V_A$} \\
    \hline
    1	& 6.0 \\
    3	& 2.0 \\
    4	& 6.0 \\
    5	& 1.0 \\
    6	& 4.0 \\
    7	& 2.0 \\
    8	& 2.0 \\
    9	& 8.0 \\
    11	& 3.0 \\
    \hline
  \end{tabular}
}
\parbox{.225\linewidth}{
    \centering
    \begin{tabular}{cc}
    \multicolumn{2}{c}{$\mathcal{T}_B$}        \\
    \hline
    \textbf{$K_B$} & \textbf{$V_B$} \\
    \hline
    2	& 1.0 \\
    4	& 5.0 \\
    5	& 1.0 \\
    8	& 2.0 \\
    10	& 4.0 \\
    11	& 2.5 \\
    12	& 6.0 \\
    15	& 6.0 \\
    16	& 3.7 \\
    \hline
    \end{tabular}
}
\parbox{.5\linewidth}{
    \centering
    \begin{tabular}{ccc}
    \multicolumn{3}{c}{$\mathcal{T}_{A \bowtie B}$}        \\
    \hline
    \textbf{$K_{A \bowtie B}$} & \textbf{$V_{A \bowtie}$} & {$V_{B\bowtie}$} \\
    \hline
    4	& 6.0	& 5.0 \\
    5	& 1.0	& 1.0 \\
    8	& 2.0	& 2.0 \\
    11	& 3.0	& 2.5 \\
    \hline \\
    \end{tabular}
    
    \texttt{SIZE}($V_{A \bowtie}$) = 4 \\
    \texttt{SUM}($V_{A \bowtie}$) = 12.0 \\
    \texttt{SUM}($V_{B \bowtie}$) = 10.5 \\
    \texttt{MEAN}($V_{A \bowtie}$) = 12.0/4 = 3.0 \\

} 
\caption{
The table $\mathcal{T}_{A \bowtie B}$ is the output of a one-to-one join between the tables $\mathcal{T}_A$ with $\mathcal{T}_B$.
We are interested in approximating post-join statistics (e.g., join size, sums, means, and covariances) of the table
$\mathcal{T}_{A \bowtie B}$ using only inner products.\vspace{-.5em}
}
\label{fig:example-tables}
\end{figure}

Sketching has become the most popular approach for performing this sort of estimation between unjoined data tables~\cite{ZhuNargesianPu:2016,FernandezMinNava:2019,YangZhangZhang:2019,SantosBessaChirigati:2021,SantosBessaMusco:2022}.
Specifically, a small-space sketch is \emph{precomputed} for all data tables in the search set. When the analyst issues a query to find relevant data, a sketch of her table is compared against these preexisting sketches using a fraction of the computational resources in comparison to explicitly materializing table joins~\cite{SantosBessaChirigati:2021}.

\myparagraph{Inner product sketching for dataset search} Interestingly, in the framework discussed above, many problems of interest can be formulated precisely as inner product sketching problems.
To see why this is the case, consider the example tables $\mathcal{T}_{A}$ and $\mathcal{T}_{B}$ shown in Figure~\ref{fig:example-tables}: each contains a  column of keys, $K_A$ and $K_B$, and a column of values,  $V_A$ and $V_B$. 
A \textit{join} operation between the tables 
on their keys 
generates the output table $\mathcal{T}_{A \bowtie B}$.\footnote{Note that, in the example described in \Cref{fig:example-tables}, we assume a one-to-one join. Dataset search problems can involve many-to-many joins as well, although a typical approach is to use a data aggregation function to reduce to the one-to-one setting \cite{SantosBessaChirigati:2021, SantosBessaMusco:2022, kanter2015deep}.}

We list in \Cref{fig:example-tables} a number of statistics that we might hope to estimate in $\mathcal{T}_{A \bowtie B}$ when searching for relevant datasets. We claim that all of these statistics can be estimated using inner products between vector representations of the tables, which we denote $\bv{x}^{\ind{K_A}}, \bv{x}^{{K_A}}$ and $\bv{x}^{\ind{K_B}}, \bv{x}^{{K_B}}$ respectively and show in  \Cref{fig:ex-vectors}.

First, it is easy to see that the size of $\mathcal{T}_{A \bowtie B}$ is equal to the intersection between the keys in $K_A$ and $K_B$, i.e., $|K_A \cap K_B| = 4$. This is in turn equal to the \emph{inner product} between $\bv{x}^{\ind{K_A}}$ and $\bv{x}^{\ind{K_B}}$. Similarly, the \texttt{SUM} aggregate of the values in $V_A$ after join (i.e., \texttt{SUM($V_{A \bowtie}$)}) is equal to the inner product $\text{\texttt{SUM($V_{A \bowtie}$)}} = \langle \bv{x}^{V_A}, \bv{x}^{\ind{K_B}} \rangle$.
To estimate a post-join mean (i.e., \texttt{MEAN}($V_{A\bowtie}$)), we can combine the join-size estimate with the SUM estimate:  
\begin{align*}
	\text{\texttt{MEAN}}(V_{A \bowtie}) = \frac{\langle \bv{x}^{V_A}, \bv{x}^{\ind{K_B}} \rangle}{\langle \bv{x}^{\ind{K_A}}, \bv{x}^{\ind{K_B}} \rangle}.
\end{align*}
Finally, computing a post-join inner product, $\langle \bv{x}^{V_A}, \bv{x}^{{V_B}} \rangle$ could be useful. 
In the application above, for tables containing precipitation data and taxi ridership, a high inner-product might signify that high precipitation days align with high ridership days.

\myparagraph{Comparison of different methods} Given the above reductions, both linear sketching methods like JL projection and CountSketch, and our Weighted MinHash method, can be directly applied to the dataset search problem. We simply need to precompute $\mathcal{S}(\bv{x}^{\ind{K_B}})$ and $\mathcal{S}(\bv{x}^{{V_B}})$ for all tables $\mathcal{T}_B$ in our search set. Sketching other vector transformations like $\mathcal{S}((\bv{x}^{{V_B}})^2)$ 
opens up the possibility of also estimating other quantities like post-join variance.

In search applications, we note that the vector length $n$ can be very large. However, computing sketches does not require fully materializing the vectors $\bv{x}^{\ind{K_A}}$ and $\bv{x}^{\ind{K_B}}$: all sketching methods discussed in this paper only need to process the vectors' non-zero entries. 
Furthermore, it is not necessary to know the $n$ beforehand: we can simply set $n$ to be large enough to cover the whole domain of the keys being sketched (e.g., $n=2^{32}$ or $n=2^{64})$.

To compare methods, \Cref{fact:jl_result} and \Cref{thm:main} suggest that any asymptotic differences in performance between our WMH method and linear sketching will depend on the overlap in non-zero entries between the vectors being sketched. In dataset search, this exactly corresponds to the \emph{Jaccard similarity} of the key sets ${K}_A$ and ${K}_B$. Our method will perform better when the Jaccard similarity is small. For example, in \Cref{fig:example-tables}, only $4$ out of $14$ unique keys are shared in both tables, so the similarity is $\approx .29$. In the scenario discussed above, we could imagine a much smaller ratio: for example, our data analyst might only have a table containing taxi data from 2022, but compare it to a weather data table with dates from 1960 through the present day. The Jaccard similarity would be $1/63 \approx . 016$. 
In Section \ref{sec:exp} we consider a dataset search use case involving data from the World Bank \cite{WBF_2022} where $42\%$ percent of table pairs had Jaccard similarity $<.1$, and $35\%$ have Jaccard similarity $< .05$.

\begin{figure}[t]
\centering
\footnotesize
\parbox{0.6\linewidth}{
\setlength{\tabcolsep}{0.35em}
\begin{tabular}{c|cccccccccccccccc}
\toprule
index & \bf 1 & \bf 2 & \bf 3 & \bf 4 & \bf 5 & \bf 6 & \bf 7 & \bf 8 & \bf 9 & \bf 10 & \bf 11 & \bf 12 & \bf 13 & \bf 14 & \bf 15 & \bf 16 \\
\midrule
$\bv{x}^{V_A}$ & 6.0 &0  & 2.0 & \bf 6.0 & \bf 1.0 & 4.0 & 2.0 & \bf 2.0 & 8.0 & 0 & \bf 3.0 & 0 & 0 & 0 & 0 & 0 \\
$\bv{x}^{\ind{K_A}}$ & 1 & 0 & 1 & \bf 1 & \bf 1 & 1 & 1 & \bf 1 & 1 & 0 & \bf 1 & 0 & 0 & 0 & 0 & 0 \\
\midrule
$\bv{x}^{V_B}$ & 0   & 1.0 &0  & \bf 5.0 & \bf 1.0 &0  &0  & \bf 2.0 & 0 & 4.0 & \bf 2.5 & 6.0 & 0 & 0 & 6.0 & 3.7 \\
$\bv{x}^{\ind{K_B}}$ &	0 & 1 & 0 & \bf 1 & \bf 1 & 0 & 0 & \bf 1 & 0 & 1 & \bf 1 & 1 & 0 & 0 & 1 & 1 \\
\bottomrule
\end{tabular}
}

\caption{
Vector representation of tables $\mathcal{T}_A$ with $\mathcal{T}_B$ from Figure~\ref{fig:example-tables}. The vector $\bv{x}^{\ind{K_A}}$ (resp. $\bv{x}^{\ind{K_B}}$) is the vector representation for the join key $K_A$ (resp. $K_B$) and $\bv{x}^{V_A}$ (resp. $\bv{x}^{V_B}$) is the vector representation for the column $V_A$ (resp. $V_B$). Bold numbers are entries included in the join result from $\mathcal{T}_{A \bowtie B}$.
}
\label{fig:ex-vectors}
\end{figure}

\subsection{Paper Roadmap}
In \Cref{sec:relate_work} we review related prior work. In \Cref{sec:unweighted-minhash} we outline an analysis of the standard \emph{unweighted} MinHash method for inner product estimation. This analysis serves as a technical warm-up for our main result (\Cref{thm:main}) on Weighted MinHash, which is presented in \Cref{sec:main_result}. Finally, in \Cref{sec:exp} we support \Cref{thm:main} with a detailed empirical evaluation of our method. 

\section{Related Work}
\label{sec:relate_work}

\myparagraph{Inner Product Estimation for Binary Vectors} Beyond linear sketching methods for estimating the inner product between general real-valued vectors $\bv{a}$ and $\bv{b}$, there has been a lot of prior work on the special case of \emph{binary} vectors with $\{0,1\}$ entries. 
For such vectors, approximating the inner product amounts to approximating the size of the intersection of two sets. Concretely, any $\bv{a}, \bv{b} \in \{0,1\}^n$ can be associated with sets $\mathcal{A}$ and $\mathcal{B}$ that contain integers from $\{1, \ldots, n\}$. We define $\mathcal{A}$ to contain all $i$ for which $\bv{a}[i] = 1$, and similarly $\mathcal{B}$ to contain all $i$ for which $\bv{b}[i] = 1$. Note that $\langle \bv{a}, \bv{b}\rangle = |\mathcal{A} \cap \mathcal{B}|$.

Applying \Cref{fact:jl_result}, we know that a linear sketch of size $m = O(1/\epsilon^2)$ can estimate $\langle \bv{a}, \bv{b}\rangle$ up to additive error $\epsilon \|\bv{a}\| \|\bv{b}\| = \epsilon \sqrt{|\mathcal{A}||\mathcal{B}|}$.
However,  a better bound can be obtained using non-linear sketching methods based on
the classic MinHash sketch \cite{Broder:1997,BroderCharikarFrieze:1998,Manber:1994,Heintze:1996}, the $k$-minimum value (KMV) sketch \cite{BeyerHaasReinwald:2007}, or related techniques \cite{LiKonig:2010,LiOwenZhang:2012}. With $m = O(1/\epsilon^2)$ space, such methods are achieve error  $\epsilon \sqrt{\max(|\mathcal{A}|, |\mathcal{B}|)\cdot |\mathcal{A}\cap \mathcal{B}|}$, which is always smaller than $\epsilon \sqrt{|\mathcal{A}||\mathcal{B}|}$ \cite{BeyerHaasReinwald:2007,PaghStockelWoodruff:2014}. For binary vectors, this bound was proven optimal in \cite{PaghStockelWoodruff:2014}.

Our work was motivated by this pre-existing result for binary vectors. In fact, our \Cref{thm:main}, is a strict generalization of the bound to \emph{all real-valued vectors}. When $\bv{a}$ and $\bv{b}$ are binary, we have that $\|\bv{a}_{\mathcal{I}}\|^2 = \|\bv{b}_{\mathcal{I}}\|^2 = |\mathcal{A}\cap \mathcal{B}|$. So it is not hard to see that $\epsilon \sqrt{\max(|\mathcal{A}|, |\mathcal{B}|)\cdot |\mathcal{A}\cap \mathcal{B}|} = \epsilon \cdot \max\left(\|\bv{a}_{\mathcal{I}}\|\|\bv{b}\|, \|\bv{a}\|\|\bv{b}_{\mathcal{I}}\| \right)$, which is exactly our bound from \Cref{thm:main}. We summarize how all prior inner product sketching methods compare to our result in \Cref{tab:error_guarantees}.

\myparagraph{Beyond Binary Vectors} 
There has been less work on obtaining better results for estimating inner products of vectors with non-binary entries. One recent paper \cite{LarsenPaghTetek:2021} proves refined bounds for the CountSketch method that depend on the $\ell_1$ norm of $\bv{a}$ and $\bv{b}$ (instead of the Euclidean norm). These bounds can be tighter than \Cref{fact:jl_result} for some vectors, especially when the sketch size $m$ is large. However,  the results are not directly comparable to ours. 

We take a different approach, moving beyond linear sketching entirely. Our main result is based on a class of sketches that we collectively refer to as ``Weighted MinHash'' methods \cite{ChiZhu:2017,Shrivastava:2016}.
These methods include weighted versions of coordinated random sampling \cite{CohenKaplan:2007,CohenKaplan:2013}, as well as the ``Consistent Weighted Sampling'' algorithm \cite{ManasseMcSherryTalwar:2010,GollapudiPanigrahy:2006} and its descendants, which are essentially equivalent, but computationally cheaper to apply \cite{Ioffe:2010,WuLiChen:2019,HaeuplerManasseTalwar:2014}. As shown in \Cref{sec:main_result}, Weighted MinHash sketches allows us to handle vectors whose entries have \emph{highly varying magnitude} (in contrast to binary vectors, where all non-zero entries have the same magnitude of $1$).

Weighted MinHash sketches have been used in a number of applications, including for approximating weighted Jaccard similarity \cite{WuLiChen:2019}, for near-duplicate detection with weighted features \cite{ManasseMcSherryTalwar:2010}, for approximating the distance between two vectors \cite{Ioffe:2010}, and for sketching image histograms \cite{Shrivastava:2016}. In many of these applications, the weighted sketches empirically outperform unweighted sketches. Weighted MinHash sketches have also been used to compute general ``sum aggregate'' queries, for which the inner product is a special case \cite{CohenKaplan:2013}. However, we are not aware of strong worst-case error guarantees for the above applications, let alone for the problem of general inner product estimation. 
Consistent Weighted Sampling has also been used to approximate inner products in \cite{kdd_2017}, albeit using a different estimator than in our work. However, non-asymptotic worst-case guarantees are not provided.

\myparagraph{Locality Sensitive Hashing} Finally, our problem of estimating inner products from sketches is closely related to cosine similarity and maximum inner product search (MIPS), where the goal is to retrieve vectors from a database with the \emph{highest} cosine similarity (respectively, inner product) with a given query vector. One approach for solving these problems is locality sensitive hashing \cite{GionisIndykMotwani:1999}, and there are methods based on both MinHash and random projections, like SimHash \cite{Charikar:2002}. It has been observed that MinHash often outperforms SimHash for binary data, which parallels what was previously known for binary  inner product estimation \cite{ShrivastavaLi:2014a}.

\section{Warmup: Unweighted MinHash} 
\label{sec:unweighted-minhash}
\label{sec:notation}
\myparagraph{Notation}
We use bold letters to denote vectors, and for a vector $\bv{a}$, $\bv{a}[k]$ denotes the $k^\text{th}$ entry (indexing starts with $1$). For two length $n$ vectors, $\bv a, \bv b$, $\langle \bv{a}, \bv{b}\rangle = \sum_{k=1}^n \bv{a}[k]\bv{b}[k]$ denotes the inner product. $\|\bv{a}\| =\ \sqrt{\langle \bv{a}, \bv{a}\rangle}$ denotes the Euclidean norm and $\|\bv{a}\|_{\infty} = \max_{k\in \{1,\ldots, n\}} |\bv{a}[k]|$ denotes the infinity norm. $\|\bv{a}\|_{1} = \sum_{k=1}^n |\bv{a}[k]|$ denotes the $\ell_1$ norm.
As is standard in the literature \cite{BeyerHaasReinwald:2007}, we assume access to uniformly random hash functions that map to the real line. I.e., we assume that we can construct a random function $h$ such that for any input $j \in \{1, \ldots, n\}$, $h(j)$ is distributed uniformly and independently on the interval $[0,1]$. In practice, $h$ can be replaced with a low-randomness function that map to a sufficient large discrete set $\{1/U, 2/U \ldots, 1\}$. Typically $U$ is chosen to equal $n^c$ for constant $c$ (e.g. $c = 3$) \cite{CormodeGarofalakisHaas:2011}.
We let $\Pr[E]$ denote the probability that a random event $E$ occurs, and $\mathbbm{1}[E]$ is the indicator random variable that evaluates to $1$ if $E$ occurs and to $0$ otherwise. $\E[X]$ and $\Var[X]$ denote the expectation and variance of a random variable $X$.

\myparagraph{An unweighted method} Before introducing our Weighted MinHash sketching method, we review the unweighted MinHash algorithm and prove a inner product estimation bound that can be obtained from this method. The bound closely follows prior work on binary vectors \cite{BeyerHaasReinwald:2007,PaghStockelWoodruff:2014} and only holds under strong assumptions on the sketched vectors $\bv{a}$ and $\bv{b}$ -- specifically that their entries are uniformly  bounded in magnitude. Nevertheless, it serves as a warmup for our main result, which is proven using a similar strategy, but eliminates the assumption by using weighted sampling.

Given a vector $\bv{a}$, we obtain an entry in the standard MinHash sketch (see e.g., \cite{Broder:1997}) by hashing the index of every non-zero entry in $\bv{a}$ to the interval $[0,1]$. We then store the smallest hash value. This process is repeated $m$ times with independently chosen random hash functions. For binary vectors $\bv{a}$ and $\bv{b}$ with non-zero index sets $\mathcal{A} = \{k: \bv{a}[k]\neq 0\}$ and $\mathcal{B} = \{k: \bv{b}[k]\neq 0\}$, the minimum hash value alone can be used to estimate the Jaccard similarity $|\mathcal{A}\cap \mathcal{B}|/|\mathcal{A}\cup \mathcal{B}|$ or the union size $|\mathcal{A}\cup \mathcal{B}|$ \cite{FlajoletMartin:1985,BeyerHaasReinwald:2007,KaneNelsonWoodruff:2010}.

\setlength{\textfloatsep}{4pt}
\begin{algorithm}[t]\caption{Unweighted MinHash Sketch}\label{alg:minhash_sketch}
	\begin{algorithmic}[1]
		\Require Length $n$ vector $\bv{a}$, sample number $m$, random seed $s$. 
		\Ensure Sketch $H_{\bv{a}} = \{H_{\bv{a}}^{hash}, H_{\bv{a}}^{val}\}$, where $H_{\bv{a}}^{hash}$ and $H_{\bv{a}}^{val}$ have length $m$ and contain values in $[0,1]$ and from $\bv{a}$, respectively
		\algrule
		\State Initialize random number generator with seed $s$.
		\For{i = 1, \ldots, m}
		\State Select uniformly random hash func. $h^i: \{1,..., n\}\rightarrow [0,1]$. 
		\State Compute $j^* = \argmin_{j \in \{1, \ldots, n\},\, \bv{a}[j] \neq 0} h^i(j)$.  \label{alg:j_star}
		\State Set $H_{\bv{a}}^{hash}[i] = h^i(j^*)$ and $H_{\bv{a}}^{val}[i] = \bv{a}[j^*]$
		\EndFor
		\State \Return $\{H_{\bv{a}}^{hash}, H_{\bv{a}}^{val}\}$
	\end{algorithmic}
 \vspace{-0.25em}
\end{algorithm}

For non-binary vectors, is it common to \emph{augment} the standard MinHash sketch by also storing the \emph{value} of the index with minimum hash value. This idea is used in ``coordinated sampling'' or ``conditional random sampling'' sketches \cite{CohenKaplan:2013,LiChurchHastie:2006,Cohen:2016}, and was recently used to extend MinHash and the closely related $k$-minimum values (KMV) sketch to estimate vector correlations \cite{SantosBessaChirigati:2021}. 
The basic augmented MinHash sketching method is shown in \Cref{alg:minhash_sketch}, which returns $H_{\bv{a}}^{hash}$ and $H_{\bv{a}}^{val}$ as vectors of minimum hashes and their corresponding vector values, respectively. 

For any single vector $\bv{a}$, the augmented MinHash sketch $H_\bv{a}$ contains a uniform subsample (collected with replacement) of the non-zero values in $\bv{a}$. This is because for all $i \in \{1, \ldots, m\}$ the minimum value of the $i^\text{th}$ hash is equally likely to come from any of the indices with non-zero value. More importantly, sketch can be used to obtain a uniform subsample from the \emph{intersection} of $\bv{a}$ and $\bv{b}$, i.e., from entries where both vectors are non-zero. This subsample can in turn be used to estimate the sum $\langle \bv{a},\bv{b}\rangle = \sum_{k=1}^n \bv{a}[k]\bv{b}[k]$, since $\bv{a}[k]\bv{b}[k]$ only contributes to the sum if $\bv{a}[k]$ and $\bv{b}[k]$ are both non-zero.
Concretely, we have the following well-known fact: 
\begin{fact}
	\label{fact:mhash_sketch}
	Consider vectors $\bv{a}$ and $\bv{b}$ sketched using \Cref{alg:minhash_sketch} to produce sketches $H_{\bv{a}}$ and $H_{\bv{b}}$. Define the sets $\mathcal{A} = \{i: \bv{a}[i]\neq 0\}$ and $\mathcal{B} = \{i: \bv{b}[i]\neq 0\}$. Then for all 
	$i\in\{1, \ldots, m\}$ we have:
	\begin{enumerate}
		\item $H_{\bv{a}}^{hash}[i] = H_{\bv{b}}^{hash}[i]$ with probability $\frac{|\mathcal{A} \cap \mathcal{B}|}{|\mathcal{A} \cup \mathcal{B}|}$.
		\item If $H_{\bv{a}}^{hash}[i] = H_{\bv{b}}^{hash}[i]$, then $H_{\bv{a}}^{val}[i] = \bv{a}[j]$ and $H_{\bv{b}}^{val}[i] = \bv{b}[j]$ for $j$ chosen uniformly at random from $\mathcal{A}\cap \mathcal{B}$.
	\end{enumerate}
\end{fact}
\Cref{fact:mhash_sketch} indicates that, to obtain a uniform subsample from the intersection of $\bv{a}$ and $\bv{b}$, we can simply take all entries in $H_{\bv{a}}^{val}$ and $H_{\bv{b}}^{val}$ where the corresponding entries in $H_{\bv{a}}^{hash}$ and $H_{\bv{b}}^{hash}$ are equal -- and, as per (1), they \emph{will be equal} with good probability. 

With \Cref{fact:mhash_sketch} in place, we describe an inner product estimator based on MinHash  (\Cref{alg:minhash_est}). This estimator will serve as a template for our weighted MinHash estimator in the next section. 

\begin{algorithm}[t]\caption{Unweighted MinHash Estimate}\label{alg:minhash_est}
	\begin{algorithmic}[1]
		\Require Sketches $H_{\bv{a}}= \{H_{\bv{a}}^{hash}, H_{\bv{a}}^{val}\}, H_{\bv{b}}= \{H_{\bv{b}}^{hash}, H_{\bv{b}}^{val}\}$ constructed using \Cref{alg:minhash_sketch} with the same inputs $m,s$.
		\Ensure Estimate of $\langle \bv{a}, \bv{b}\rangle$.
		\algrule
		\State Set $\tilde{U}= \frac{m}{\sum_{i=1}^m\min\left( H_{\bv{a}}^{hash}[i],H_{\bv{b}}^{hash}[i]\right)} - 1$
		\State \Return $\frac{\tilde{U}}{m} \sum_{i=1}^m \mathbbm{1}\left[ H_{\bv{a}}^{hash}[i] =H_{\bv{b}}^{hash}[i]\right] \cdot H_{\bv{a}}^{val}[i]\cdot H_{\bv{b}}^{val}[i]$ \label{alg:minhash_est_summation}
	\end{algorithmic}
\end{algorithm}
Consider the summation in line~\ref{alg:minhash_est_summation} of \Cref{alg:minhash_est}.
Using linearity of expectation and \cref{fact:mhash_sketch}, we can compute the expectation:
\begin{align*}
	&\E\left[\sum_{i=1}^m \mathbbm{1}\left[ H_{\bv{a}}^{hash}[i]=H_{\bv{b}}^{hash}[i]\right] \cdot H_{\bv{a}}^{val}[i]\cdot H_{\bv{b}}^{val}[i]\right]  \\&= m\cdot \E\left[\mathbbm{1}\left[ H_{\bv{a}}^{hash}[1] =H_{\bv{b}}^{hash}[1]\right] \cdot H_{\bv{a}}^{val}[1]\cdot H_{\bv{b}}^{val}[1]\right]\\ &= m\cdot  \sum_{j\in \mathcal{A}\cap \mathcal{B}}\frac{1}{|\mathcal{A}\cup \mathcal{B}|} \bv{a}[j]\bv{b}[j] = \frac{m}{|\mathcal{A}\cup \mathcal{B}|}\cdot \langle \bv{a}, \bv{b}\rangle. 
\end{align*}

It follows from the above that, if we multiplied the summation $\sum_{i=1}^m \mathbbm{1}\left[ H_{\bv{a}}^{hash}[i] =H_{\bv{b}}^{hash}[i]\right] \cdot H_{\bv{a}}^{val}[i]\cdot H_{\bv{b}}^{val}[i]$ by $\frac{|\mathcal{A}\cup \mathcal{B}|}{m}$, then we would have an unbiased estimate for $\langle \bv{a}, \bv{b}\rangle$, as desired. 
The only catch is that we do not \emph{know} $|\mathcal{A}\cup \mathcal{B}|$. This union size cannot be computed exactly from our sketches $H_\bv{a}$ and $H_{\bv{b}}$. However, it can be \emph{estimated} using the same information contained in our MinHash sketches. In particular, since $h^i$ hashes uniformly to $[0,1]$,  $\frac{m}{\sum_{i=1}^m\min\left( H_{\bv{a}}^{hash}[i],H_{\bv{b}}^{hash}[i]\right)} - 1$ provides a good estimate for $|\mathcal{A}\cup \mathcal{B}|$. This is  actually a standard variant of the well-known Flajolet-Martin distinct elements estimator \cite{FlajoletMartin:1985,BeyerHaasReinwald:2007}. In Line 1 of \cref{alg:minhash_est}, we set $\tilde{U}$ equal to 
this estimator 
and we multiply by $\frac{\tilde{U}}{m}$ in Line 2 as a surrogate for $\frac{|\mathcal{A}\cup \mathcal{B}|}{m}$.  This gives our final estimator for $\langle \bv{a}, \bv{b}\rangle$.

Overall, we are able to prove the following concentration bound for the estimator for computing the inner product between any pair of \emph{bounded vectors}. For binary vectors, the constant $c$ below equals $1$ and we exactly recover the bounds from prior work \cite{PaghStockelWoodruff:2014}.

\begin{theorem}[Intermediate Result: Inner Product Sketching with Unweighted MinHash]
	\label{thm:bounded}
	Let $\epsilon,\delta \in (0,1)$ be accuracy and failure probability parameters and let $m = O(\log(1/\delta)/\epsilon^2)$.  There is an algorithm $\mathcal{S}$ that produces size-$m$ sketches (\Cref{alg:minhash_sketch}),  along with an estimation procedure $\mathcal{F}$, such that for any $\bv a,\bv b \in \R^n$ with entries bounded in $[-c,c]$, with probability at least $1-\delta$,
	\begin{align*}
		\left | \mathcal{F}\left(\mathcal{S}(\bv{a}), \mathcal{S}(\bv{b})\right) - \langle \bv{a}, \bv{b} \rangle\right | \leq  \epsilon \cdot c^2 \cdot \sqrt{\max(|\mathcal{A}|, |\mathcal{B}|)\cdot |\mathcal{A}\cap \mathcal{B}|}
	\end{align*}
	for $\mathcal{A} = \{i: \bv{a}[i] \neq 0\}$ and $\mathcal{B} = \{i: \bv{b}[i] \neq 0\}$.
\end{theorem}

The full proof of \Cref{thm:bounded} is included in \Cref{app:unweighted_proof}. It requires two technical ingredients. First, we must bound the variance of an ``ideal'' estimator that uses the exact value of $|\mathcal{A}\cup \mathcal{B}|$. This can be done by using the fact that  $\bv a$ and $\bv b$ have entries bounded in $[-c,c]$. Second, we can bound the error introduced by replacing $|\mathcal{A}\cup \mathcal{B}|$ with an estimate for the union, as discussed above. To do so, we rely on the following standard result, which shows that MinHash sketches for $\bv{a}$ and $\bv{b}$ can be used to compute a $(1\pm\epsilon)$ relative error approximation to the true union $|\mathcal{A}\cup \mathcal{B}|$ when $m = O(1/\epsilon^2)$:
\begin{lemma}[Union Size Estimator \cite{BlumHopcroftKannan:2020}] 
	\label{lem:distinct_elem}
	Let $\mathcal{A}$ and $\mathcal{B}$ be non-empty subsets of $\{1, \ldots, n\}$ and let $h^1, \ldots, h^m: \{1, \ldots, n\}\rightarrow [0,1]$ be independent, uniform random hash functions. 
For any $\epsilon, \delta \in (0,1)$, if $m = O \left (\frac{1}{\delta\epsilon^2}\right )$, then with prob. at least $1-\delta$, the estimator $\tilde{U}= \frac{m}{\sum_{i=1}^m\min_{j\in \mathcal{A}\cup\mathcal{B}} h^i(j)} - 1$ satisfies:
	\begin{align*}
		(1-\epsilon)|\mathcal{A}\cup \mathcal{B}| \leq  \tilde{U} \leq 	(1+\epsilon)|\mathcal{A}\cup \mathcal{B}|.
	\end{align*}
\end{lemma}
Note that, while it is written in a slightly different way, the $\tilde{U}$ in \Cref{lem:distinct_elem} is exactly equivalent to the $\tilde{U}$ in \Cref{alg:minhash_est}  (when $\mathcal{A}$ and $\mathcal{B}$ contain the non-zero indices of $\bv{a}$ and $\bv{b}$). To see why this is the case, note that $H_{\bv{a}}^{hash}[i] = \min_{j\in \mathcal{A}} h^i(j)$ and $H_{\bv{b}}^{hash}[i] = \min_{j\in \mathcal{B}} h^i(j)$. So $\min\left( H_{\bv{a}}^{hash}[i],H_{\bv{b}}^{hash}[i]\right) = \min_{j\in \mathcal{A}\cup\mathcal{B}} h^i(j)$. 

\section{Main Result: Weighted MinHash}
\label{sec:main_result}
The main technical challenge in our work is extending the results of the previous section (\Cref{thm:bounded}) to vectors whose entries have \emph{highly varying magnitude}. It is not hard to see that the simple MinHash method fails for such vectors. For example, consider the extreme case when $\bv{a}$ and $\bv{b}$ both contain a very large values at some index $i$, so large that the term $\bv{a}[i]\bv{b}[i]$ dominates the inner product $\langle \bv{a},\bv{b}\rangle = \sum_{k=1}^n \bv{a}[k]\bv{b}[k]$. To correctly approximate the inner product, we \emph{need} to include $\bv{a}[i]$ and $\bv{b}[i]$ in our sketches for $\bv{a}$ and $\bv{b}$, respectively. A MinHash sketch will only do so with low probability, since it uniformly samples entries from the intersection of the vectors. Thus, it will obtain a poor estimate for $\langle \bv{a},\bv{b}\rangle$.

To address the issue with \emph{heavy} entries, we modify the approach of \Cref{sec:unweighted-minhash} to incorporate \emph{non-uniform} sampling weights using a Weighted MinHash sketch \cite{ManasseMcSherryTalwar:2010}. This allows us to sample high magnitude entries in the vectors with higher probability. Specifically, our goal is to sample the $i^\text{th}$ entry of $\bv{a}$ with probability proportional to the squared magnitude, $\bv{a}[i]^2$.
The Weighted MinHash sketch achieves non-uniform sampling in a simple way: we construct an \emph{extended} vector $\bv{\bar{a}}$ which has the same entries as $\bv{a}$, but entries are repeated multiple times, with the exact number of repetitions proportional to their magnitude. We then apply the standard MinHash sketch to $\bv{\bar a}$. This approach is detailed  in \Cref{alg:weighted_sketch}. 

\setlength{\textfloatsep}{4pt}
\begin{algorithm}[t]\caption{Weighted MinHash Sketch}\label{alg:weighted_sketch}
	\begin{algorithmic}[1]
		\Require Length $n$ vector $\bv{a}$, sample number $m$, random seed $s$, integer discretization parameter $L$.
		\Ensure Sketch $W_{\bv{a}} = \{W_{\bv{a}}^{hash}, W_{\bv{a}}^{val}, \|\bv{a}\|\}$, where $W_{\bv{a}}^{hash}$ is a length $m$ vector of values in $[0,1]$, $W_{\bv{a}}^{val}$ is a length $m$ vector containing a subset of entries from $\bv{a}$, and $\|\bv{a}\|$ is a scalar, the Euclidean norm of $\bv a$.
		\algrule
		\State Initialize random number generator with seed $s$.
		\State Set $\tilde{\bv{a}} =  \textsc{Round}(\bv{a}/\|\bv{a}\|$, L) using \cref{alg:round}.
		\State For each $i \in \{1, \ldots, n\}$, let $\bar{\bv{a}}^{(i)}$ be a length $L$ vector whose first $\tilde{\bv{a}}[i]^2 \cdot L$  entries are set to $\tilde{\bv a}[i]$. Set the remaining entries  to $0$.
		\State Let $\bv{\bar{a}} = [\bar{\bv{a}}^{(1)}, \ldots, \bar{\bv{a}}^{(n)}]$ be a length $n\cdot L$ vector obtained by concatenating the vectors defined above.
		
		\For{i = 1, \ldots, m}
		\State Select uniform random hash func. $h^i: \{1,..., nL\}\rightarrow [0,1]$. 
		\State Compute $j^* = \argmin_{j \in \{1, \ldots, n\cdot L\},\, \bar{\bv{a}}[j] \neq 0} h^i(j)$. 
		\State Set $W_{\bv{a}}^{hash}[i] = h^i(j^*)$ and $W_{\bv{a}}^{val}[i] = \bar{\bv{a}}[j^*]$. 
		\EndFor
		\State \Return $\{W_{\bv{a}}^{hash}, W_{\bv{a}}^{val}, \|\bv{a}\|\}$.
	\end{algorithmic}
    \vspace{-.25em}
\end{algorithm}
\begin{algorithm}[t]\caption{Vector Rounding for Weighted MinHash}
    \label{alg:round}
    \vspace{-.25em}
	\begin{algorithmic}[1]
		\Require Length $n$ unit vector $\bv{z}$, integer discretization parameter $L$.
		\Ensure Length $n$ unit vector $\bv{\tilde z}$ with $\bv{\tilde z}[i]^2$ an integer multiple of $1/L$ for all $i$.
		\algrule
		\State For all $i \in \{1, \ldots, n \}$, $\bv{\tilde z}[i] = \sign(\bv{z}[i]) \cdot \sqrt{\frac{\lfloor {\bv{z}}[i]^2 \cdot L \rfloor}{L}}$.
		\State Let $i^* = \argmax_{i \in 1,\ldots,n} |\bv{z}[i]|$. 
		\State Fix $\delta = 1-\norm{\bv{\tilde z}}^2$, then set $\bv{\tilde z}[i^*] = \sign(\bv z[i^*]) \cdot \sqrt{\bv{\tilde z}[i^*]^2  + \delta}$.
		\State \Return $\bv{\tilde z}$.
	\end{algorithmic}
\end{algorithm}

\myparagraph{Rounding \& Normalization} While Weighted MinHash allows us to  sample entries with non-uniform probability, another challenge arises: since sketches for $\bv{a}$ and $\bv{b}$ are computed independently, we no longer sample with the \emph{same probability} from both vectors. For $\bv{b}$, Weighted MinHash samples indices with probability proportional to $\bv{b}[i]^2$ instead of $\bv{a}[i]^2$. This mismatch can actually \emph{reduce} the probability that we select entries from $\bv{a}$ and $\bv{b}$ with the same index.

We are able to balance this issue with a normalization strategy. In particular, line 2 in \Cref{alg:weighted_sketch} performs a simple but important preprocessing step that \emph{scales} and \emph{rounds} $\bv{a}$ to a unit vector $\bv{\tilde{a}}$ whose squared entries are all integer multiples of $1/L$ for some large integer $L$ (to be chosen later). The rounding handles a minor issue: since we control the frequency with which each entry $\bv{a}[i]$ is sampled by \emph{repetition}, we need the squared value of all entries to be integer multiples of the same fixed constant in order to sample precisely with probability proportional to $\bv{a}[i]^2$. As will be proven, $L$ can be chosen so that the discretization has little impact on the accuracy of our final inner product estimate, and the parameter also has no impact on the size of the sketch returned by \Cref{alg:minhash_sketch}.\footnote{Note that our rounding method (\Cref{alg:round}) is non-standard: It rounds all entries of the input vector {down} to smaller magnitude values, except for the \emph{largest} magnitude entry in the vector, which gets rounded up. 
This scheme allows us to achieve small \emph{relative error} when rounding
and to avoid additive error depending on $1/L$.}

The scaling is what deals with the bigger issue discussed above, which is the mismatch in sampling probabilities between $\bv{a}$ and $\bv{b}$. Surprisingly, we can show that the impact of this mismatch can be controlled when $\|\bv{a}\| = \|\bv{b}\|$. So while it is possible to come up with examples where the algorithm fails if we directly sketch $\bv{a}$ and $\bv{b}$, we can obtain a worst-case bound by sketching $\bv{a}/\|\bv{a}\|$ and $\bv{b}/\|\bv{b}\|$, approximating $\langle \bv{a}/\|\bv{a}\|, \bv{b}/\|\bv{b}\| \rangle$, and then post-multiplying the result by $\|\bv{a}\|\|\bv{b}\|$ to get our final estimator.

\begin{algorithm}[t]\caption{Weighted MinHash Estimate}\label{alg:weight_est}
	\begin{algorithmic}[1]
		\Require Sketches $W_{\bv{a}} = \{W_{\bv{a}}^{hash}, W_{\bv{a}}^{val}, \|\bv{a}\|\}$ and $W_{\bv{b}} = \{W_{\bv{b}}^{hash}, W_{\bv{b}}^{val}, \|\bv{b}\|\}$ constructed using \Cref{alg:weighted_sketch} with the same inputs $m$, $s$, and  $L$.
		\Ensure Estimate of $\langle \bv{a}, \bv{b}\rangle$.
		\algrule
		\State For $i \in \{1, \ldots, m \}$, set $q_i = \min\left(W_{\bv{a}}^{val}[i]^2,W_{\bv{b}}^{val}[i]^2\right)$.
		\State Set $\tilde{M} = \frac{1}{L}\cdot  \left(\frac{m}{\sum_{i=1}^m\min\left(W_{\bv{a}}^{hash}[i],W_{\bv{b}}^{hash}[i]\right)} - 1\right)$.
		\State Set $I = \frac{\tilde{M}}{m} \sum_{i=1}^m \mathbbm{1}\left[W_{\bv{a}}^{hash}[i]  = W_{\bv{b}}^{hash}[i]\right]\cdot \frac{W_{\bv{a}}^{val}[i] \cdot W_{\bv{b}}^{val}[i] }{q_i}.$
		\State \Return $\|\bv{a}\|\|\bv{b}\|\cdot I$ \label{alg_weight_est_return}
	\end{algorithmic}
\end{algorithm} 

\myparagraph{Deriving the Inner Product Estimator} 
We next motivate \Cref{alg:minhash_est}, which is the algorithm used to estimate $\langle \bv{a}, \bv{b}\rangle$ from our sketches.
Note that Weighted MinHash Sketch (\cref{alg:weighted_sketch}) in fact returns an Unweighted MinHash Sketch  (\cref{alg:minhash_sketch}) for the expanded vectors $\bv{\bar a}, \bv{\bar b}$. So, we can apply \Cref{fact:mhash_sketch} to obtain the following:
\begin{fact}
	\label{fact:mhash_sketchWeighted}
	Consider vectors $\bv{a}$ and $\bv{b}$ sketched using \Cref{alg:weighted_sketch} to produce $W_{\bv{a}}$ and $W_{\bv{b}}$. Define $\mathcal{A}$ and $\mathcal{B}$ as in \Cref{fact:mhash_sketch}. For all $i\in\{1, \ldots, m\}$ we have:
	\begin{enumerate}
		\item $W_{\bv{a}}^{hash}[i] = W_{\bv{b}}^{hash}[i]$ with probability equal to the 
		weighted Jaccard similarity, $\bar{J} = \frac{\sum_{j=1}^n \min(\tilde{\bv a}[j]^2,\tilde{\bv b}[J]^2)}{\sum_{j =1}^n \max(\tilde{\bv a}[j]^2,\tilde{\bv b}[j]^2)}$.
		\item If $W_{\bv{a}}^{hash}[i] = W_{\bv{b}}^{hash}[i]$, then we have that $W_{\bv a}^{val}= \bv{\tilde a}[j]$ and  $W_{\bv b}^{val}= \bv{\tilde b}[j]$ for $j$ chosen from $\mathcal{A} \cap \mathcal{B}$ with probability equal to 
		$\min(\tilde{\bv a}[j]^2,\tilde{\bv b}[j]^2)/\sum_{i=1}^{n} \max(\tilde{\bv a}[j]^2,\tilde{\bv b}[j]^2)$.
	\end{enumerate}
\end{fact}
A proof of \Cref{fact:mhash_sketchWeighted} is given in \cref{app:weighted_proof}. With the statement in place, we present our procedure  for estimating $\langle \bv{a}, \bv{b}\rangle$ based the sketches computed by \Cref{alg:weighted_sketch}. This procedure, shown in \Cref{alg:weight_est}, is reminiscent of our estimator for unweighted sketches from the previous section. 
The only difference is that, since we are sampling with non-uniform probabilities, we need to inversely weight samples in our sum to keep everything correct in expectation.
In particular, consider the sum 
in line 3 of the algorithm.\\
By \cref{fact:mhash_sketchWeighted} and linearity of expectation, we have that:
\vspace{-.25em}
\begin{align*}
	&\E\left[\sum_{i=1}^m \mathbbm{1}\left[W_{\bv{a}}^{hash}[i]  = W_{\bv{b}}^{hash}[i]\right]\cdot \frac{W_{\bv{a}}^{val}[i] \cdot W_{\bv{b}}^{val}[i] }{q_i}\right]  \\&
	m\cdot \E\left[\mathbbm{1}\left[W_{\bv{a}}^{hash}[i]  = W_{\bv{b}}^{hash}[i]\right]\right]\cdot \frac{W_{\bv{a}}^{val}[i] \cdot W_{\bv{b}}^{val}[i] }{q_i}\\
	&= m\cdot  \sum_{j\in \mathcal{A}\cap \mathcal{B}} \frac{q_j}{\sum_{i=1}^n \max(\tilde{\bv a}[i]^2,\tilde{\bv b}[i]^2)} \frac{\tilde{\bv a}[j]\tilde{\bv b}[j]}{q_j} \\
	&= \frac{m}{\sum_{i=1}^n \max(\tilde{\bv a}[i]^2,\tilde{\bv b}[i]^2)}\cdot \langle \tilde{\bv a}, \tilde{\bv b}\rangle. 
\end{align*}
\vspace{-.25em}
So, we have obtained an estimator that in expectation is equal to $\langle \tilde{\bv a}, \tilde{\bv b}\rangle$, multiplied by $m$ over a term $M =\sum_{i=1}^n \max(\tilde{\bv a}[j]^2,\tilde{\bv b}[j]^2)$. This term $M$ is referred to as the  \emph{weighted} union size between the vectors. We can multiply by $\frac{M}{m}$ to obtain an unbiased estimator for $\langle \tilde{\bv a}, \tilde{\bv b}\rangle$. Since $\tilde{\bv{a}}$ and $\tilde{\bv{b}}$ were obtained by scaling $\bv{a}$ and $\bv{b}$ inversely by their Euclidean norms (ignoring the effect of rounding for now), our final estimator in Line~\ref{alg_weight_est_return} of \Cref{alg:weight_est} multiplies by $\|\bv{a}\|\|\bv{b}\|$. The values of $\|\bv{a}\|$ and $\|\bv{b}\|$ are stored explicitly in the sketches for $\bv{a}$ and $\bv{b}$, respectively (as just one extra number per sketch).

The formal analysis of \Cref{alg:weight_est}, which yields \Cref{thm:main}, is included in \Cref{app:weighted_proof}.
It contains three parts. First, when analyzing the unweighted estimator, we do not know $M$ exactly, so must estimate it. We can take advantage of the fact that $M$ is exactly equal to the \emph{unweighted} union size $|\bar{\mathcal{A}} \cup \bar{\mathcal{B}}|$ between the non-zero index sets  $\bar{\mathcal{A}}$ and $\bar{\mathcal{B}}$ of the \emph{expanded vectors} $\bar{\bv{a}}$ and $\bar{\bv{b}}$ constructed in \Cref{alg:weighted_sketch}. We can apply \Cref{lem:distinct_elem} directly to obtain an estimator, which is denoted as $\tilde{M}$ in \Cref{alg:weight_est}. 
Second, we need to analyze the variance of the sum $\sum_{i=1}^m \mathbbm{1}\left[W_{\bv{a}}^{hash}[i]  = W_{\bv{b}}^{hash}[i]\right]\cdot \frac{W_{\bv{a}}^{val}[i] \cdot W_{\bv{b}}^{val}[i] }{q_i}$. This analysis
uses the fact that 
$\tilde{\bv{a}}$ and $\tilde{\bv{b}}$ are unit vectors. Third, we need to rigorously analyze the impact of the rounding procedure performed in Line 2 of \Cref{alg:weighted_sketch} to establish that a good estimate for $\langle \tilde{\bv a}, \tilde{\bv b}\rangle$ actually yields a good estimate for $\langle \bv{a}/\|\bv{a}\|, \bv{b}/\|\bv{b}\|\rangle = \frac{1}{\|\bv{a}\|\|\bv{b}\|} \langle \bv{a}, \bv{b}\rangle$.

We conclude by noting that our final analysis of \Cref{alg:weight_est} requires setting $L$ to be on the order of $n^6/\epsilon^2$ when sketching using \Cref{alg:weighted_sketch}. This may sound large, but note that the parameter has \emph{no impact} on the size of the sketches returned by \Cref{alg:weighted_sketch}, or on the runtime of our estimation procedure \Cref{alg:weight_est}. $L$ does impact the {runtime} of  \Cref{alg:weighted_sketch}, but as discussed in \cref{sec:exp}, prior work can be used to implement the Weighted MinHash sketching method so that it has a \emph{logarithmic} dependence on $L$ -- i.e., on $O(\log(n/\epsilon))$.

\section{Experiments}
\label{sec:exp}
To support the results presented in \cref{sec:main_result},
we performed an experimental evaluation using synthetic data and real-world datasets.

\myparagraph{Baselines}
We compare our Weighted MinHash approach against 4 baseline methods, 2 linear and 2 sampling-based, with the goal of evaluating the trade-off between sketch size and accuracy in estimating inner products.  Those methods are:
\smallskip \vspace{-.5em}
\begin{description}[style=unboxed,leftmargin=0cm]
	\item[\em Johnson-Lindenstrauss Projection (JL):] equivalent to the AMS  sketch \cite{AlonMatiasSzegedy:1999,Achlioptas:2003}. Uses a random matrix $\bs{\Pi}$ with scaled $\pm 1$ entries (\Cref{fact:jl_result}).
	\item[\em CountSketch (CS):] classic linear sketch introduced in \cite{CharikarChenFarach-Colton:2002}, and corresponds to multiplication with a $\bs{\Pi}$ that has sparse random entries. We follow the implementation in \cite{LarsenPaghTetek:2021}, using $5$ repetitions of the sketch and taking the median to improve performance.
	\item[\em MinHash Sampling (MH):] method described in \Cref{alg:minhash_sketch}; we use a single sketch without any median estimate.  
	\item[\em $k$-Minimum Values Sampling (KMV):] sampling-based sketch closely related to MinHash, but it draws samples from the vector being sketched \emph{without replacement}. It can also be used to estimate union size. We follow the implementations from \cite{BeyerHaasReinwald:2007} and \cite{SantosBessaChirigati:2021}.
	\item[\em Weighted MinHash Sampling (WMH):] our method described in \Cref{alg:weighted_sketch}; we use a single sketch without any median estimate. 
\end{description}
\vspace{-.25em}
\myparagraph{Storage Size}
For linear sketches, we store the output of the matrix multiplication $\bs{\Pi}\bv{a}$ as 64-bit doubles. We also store $W_\bv{a}^{val}$ and $H_\bv{a}^{val}$ as 64-bit doubles. Since sampling-based sketches need to store hash values (which in our case are 32-bit ints), a sampling-based sketch with $m$ samples takes $1.5x$ as much space as a JL sketch with $m$ rows.  In our experiments, we plot \emph{storage size} which denotes the total number of bits in the sketch divided by 64,  i.e., the total number of 64-bit doubles (or equivalent) used in the sketch. Standard quantization tricks could likely be used to reduce the size of numbers in all sketches (linear and sampling), but we leave the development of such methods to future work. As a starting point, we note that there has already been interesting work on quantized JL projections \cite{Jacques:2015, LiMitzenmacherSlawski:2016}, and the SimHash method for estimating cosine similarity can be viewed as a ``1-bit'' quantization of a JL sketch \cite{Charikar:2002}. 

\myparagraph{Estimation Error} For all plots, we report the absolute difference between $\langle \bv{a}, \bv{b}\rangle$ and the estimate, divided by $\|\bv{a}\|\|\bv{b}\|$. This is the term appearing on the right-hand side of the accuracy guarantee for linear sketches \Cref{fact:jl_result}, so this scaling roughly ensures that errors are between $0$ and $1$,
making it easier to compare across different datasets. We always report average error over $10$ independent trials. 

\myparagraph{Choice of $L$} 
Note that the choice of $L$ in \Cref{alg:weighted_sketch} does not impact the size of our final sketch, so in general, it should be set as large as possible. Our bounds from \Cref{lem:round} that suggest $L$ should be set $\geq n^6$ are likely loose (we did not attempt to optimize polynomial factors), but we did find that it \emph{is necessary} to at least ensure that $L > n$. Ideally it should be larger by a multiplicative factor 100 or 1000. The reason for this is that, if $\bv{a}$ is dense and is normalized to have unit norm, as in \Cref{alg:minhash_sketch},  \emph{most} of its entries could have squared value $< 1/n$ (as the average value of a squared entry in a unit norm vector is always $1/n$). If we set $L < 1/n$, then any entries with value $<1/n$ would get rounded to $0$, which could negatively impact the accuracy of an inner product estimate.

\myparagraph{Efficient Weighted Hashing} 
When $L$ is large, a {naive} implementation of \Cref{alg:weighted_sketch} would be prohibitively slow. The ``extended'' vector $\bar{\bv{a}}$ has length $n\cdot L$ and we must apply a hash function to every non-zero entry in that vector. Let $\mathcal{A} = \{i: \bv{a}[i]\neq 0\}$ as before, so $|\mathcal{A}|$ is equal to the number of non-zero values in $\bv{a}$.  If each hash computation is considered unit cost, this amounts to a runtime of $O(|\mathcal{A}|m\cdot L)$, which is too large, since $L$ is chosen larger than $n$.

Fortunately, it is possible to improve this cost to $O(|\mathcal{A}|m\cdot \log L) = O(|\mathcal{A}|m\cdot \log n) $ using  techniques for speeding up weighted MinHash sketches. Such techniques have been heavily studied in recent years \cite{Ioffe:2010,WuLiChen:2019,HaeuplerManasseTalwar:2014,Shrivastava:2016}. The savings are significant, reducing the computation cost of sketching to nearly-linear in the size of the input for each of our $m$ samples. 
Among faster methods, we specifically  employ the simple ``active index''  technique, which was first introduced in \cite{GollapudiPanigrahy:2006}. The rough idea is that, when hashing non-zero entries in a particular length $L$ block of $\bv{\bar{a}}$, there is no need to hash all non-zero indices in that block. We can skip over large sections of indices by observing that if $z$ is the minimum hash value generated so far, the next index where a lower hash value will be seen is a distributed as a \emph{geometric random variable} with parameter $z$. We can sample from the geometric distribution efficiently (e.g. using a built-in Python routine) and skip ahead to that index. It is possible to prove that the expected cost of this approach is just $O(\log L)$ per block. See the exposition in \cite{ManasseMcSherryTalwar:2010} 
for further details.

Since initially releasing this paper, we became aware of even faster implementations of weighted MinHash that reduce the runtime to $O(|\mathcal{A}| + m\log m)$, which is nearly linear in the number of non-zeros in the vector being sketched \cite{ertl2018bagminhash,christiani2020dartminhash}. Such methods should be able to be adapted for use in our inner product sketching application, although we leave further exploration to future work.

\myparagraph{Choice of Hash Function} In practice we cannot obtain a truly uniform random hash function from $\{1,\ldots, n\}$ to the reals, so we must use an approximation. In our experiments, we employ a standard 2-wise independent hash function (linear function with random coefficients) that maps from $\{1, \ldots, n\}$ to $\{1, \ldots, p\}$ for a 31-bit prime $p$ \cite{CarterWegman:1979}. \footnote{Our choice to use a 2-wise independent hash function was based on prior implementations of the weighted MinHash method \cite{WuLiChen:2020} that do so.} We then use as our hash value $h(i)/p$, which is a number between $0$ and $1$. Since $p$ is chosen to have 31 bits, we can store the value of $h(i)$ in our sketch using a standard 32-bit int.

\begin{figure}[t]
	\centering
	
	\begin{subfigure}{.8\columnwidth} 
		\centering
		\includegraphics[width=1.0\linewidth]{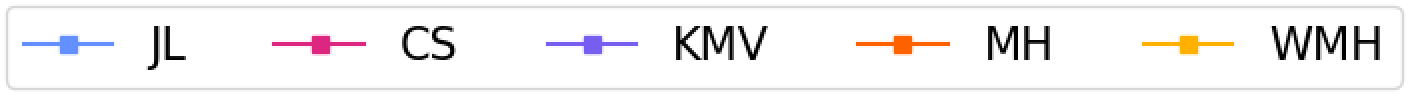}
	\end{subfigure}
	
	\begin{subfigure}{.49\columnwidth} 
		\centering
		\includegraphics[width=1.0\linewidth]{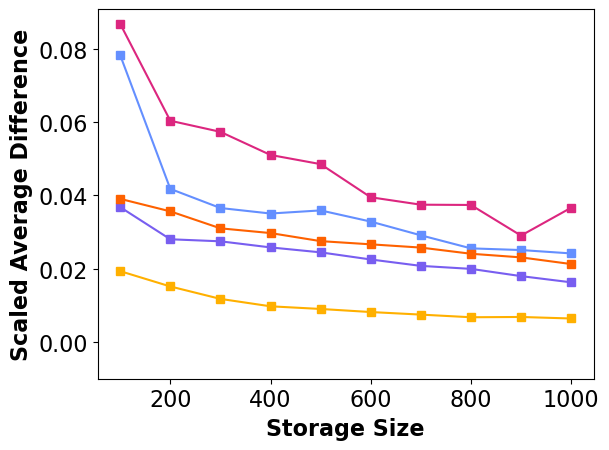}
		\vspace{-1.8em}
		
		\caption{1\% overlap}
	\end{subfigure}
	\begin{subfigure}{.49\columnwidth} 
		\centering
		\includegraphics[width=1.0\linewidth]{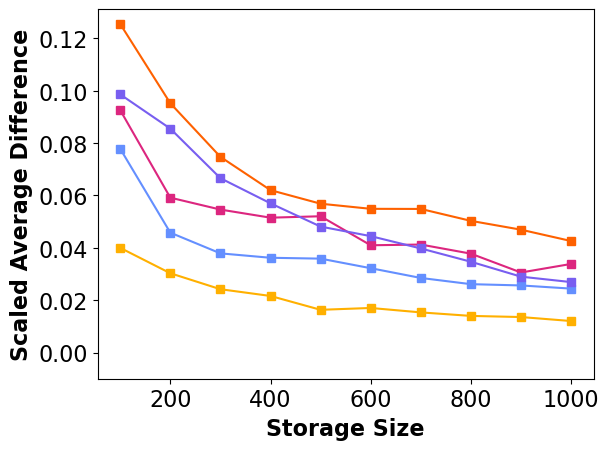}
		\vspace{-1.8em}
				
		\caption{5\% overlap}
	\end{subfigure}%
	
	\begin{subfigure}{.49\columnwidth} 
		\centering
		\includegraphics[width=1.0\linewidth]{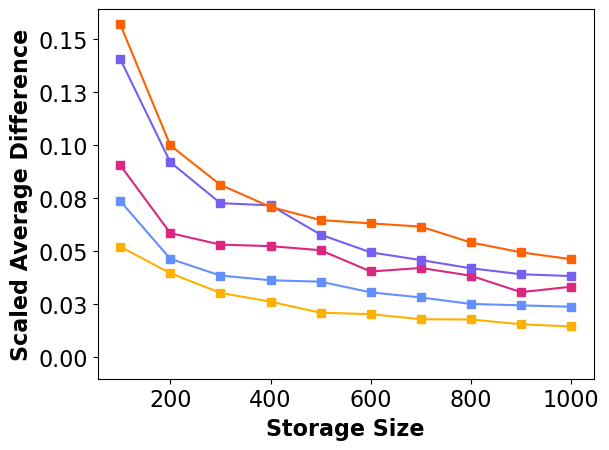}
				\vspace{-1.8em}
				
		\caption{10\% overlap}
	\end{subfigure}
	\begin{subfigure}{.49\columnwidth} 
		\centering
		\includegraphics[width=1.0\linewidth]{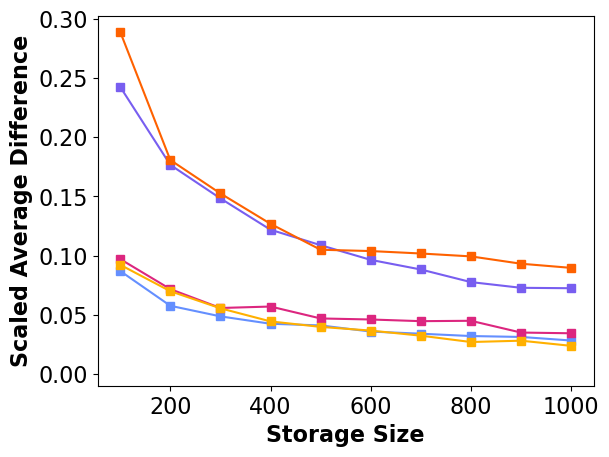}
				\vspace{-1.8em}
				
		\caption{50\% overlap}
	\end{subfigure}
    
	\caption{Inner product estimation (synthetic data).}
	\vspace{-.1em}
	\label{fig:InnerProduct_Corr098_Est}
\end{figure}


\vspace{-.2cm}
\subsection{Synthetic Data}
We begin with an evaluation of our approach using synthetic data.
%
We generate length $10000$ vectors $\bv{a}$ and $\bv{b}$, each with $2000$ non-zero entries. The ratio of non-zero entries that \emph{overlap}, i.e., are non-zero in both $\bv{a}$ and $\bv{b}$, is adjusted to simulate different practical settings with different levels of \emph{joinability} between tables (see \Cref{sec:apps}). The non-zero entries in $\bv{a}$ and $\bv{b}$ are normal random variables with values between $-1$ and $1$, 
except $10\%$ of entries are chosen randomly as outliers and set to random values between $20$ and $30$. 

Results for varying amounts of overlap are reported in \Cref{fig:InnerProduct_Corr098_Est}. They closely align with our theoretical findings: when the overlap is small, the bounds for 
Weighted MinHash
are significantly better than those of linear sketching methods. Accordingly, WMH outperforms all other methods for overlap ratio $\leq 10\%$. Note that unweighted sampling based sketches also outperform linear sketches for very low overlap ($1\%$). But as the overlap increases, the advantage brought about by \Cref{thm:main} over \Cref{fact:jl_result} decreases.  We can see this in \Cref{fig:InnerProduct_Corr098_Est}(d): at 50\% overlap, the performance of linear sketching 
is comparable to that of Weighted MinHash. 

\subsection{Real-World Data}

\myparagraph{
Assessing the Effect of Overlap and Outliers}
Using sketches of size 400,\footnote{The size was chosen empirically.
Our goal here is to simulate the real-world situation where a fixed parameter must be selected for a given application.} we estimate the inner product between 5000 pairs of numerical columns from 56 datasets published by the World Bank Group~\cite{WBF_2022}. 
We normalize columns to have norm $1$ so that all inner products have magnitude less than $1$.  
We visual results using
a winning table in \Cref{fig:real}, filting vector pairs based on different overlap ratios (column) and kurtosis values, a measure of outliers (row). Each cell shows the average error difference (WMH estimation error minus the error of other method) for vector pairs with the specified overlap and kurtosis values.

The blue cells (negative difference) correspond to combinations in which 
WMH outperforms the other methods, while the red cells (positive difference) represent combinations in which the other methods win. The darker the cells, the bigger the difference.
A high kurtosis often indicates the presence of outliers, which will, based on our theoretical results, 
present a difficulty for unweighted sampling methods like MH in comparison to JL or our WMH method. This  is supported by the experiments, which show that WMH has a great improvement over MH when kurtosis is high (up to -.031 vs. at most -.020 when kurtosis is low). 
As predicted by \Cref{thm:main} and shown in our synthetic experiments, WMH also has a great edge over JL for low overlap values. For large overlaps (greater than .75), JL leads to slightly smaller errors (from 0.003 to 0.006). 

This suggests that \emph{WMH provides a good compromise for applications in which the distribution of data is unknown: it provides much better estimates for many cases, and when it does not, its estimates are comparable to 
the best results from existing sketching methods}.

\begin{figure}[t]
	\centering
	\begin{subfigure}{0.48\columnwidth} 
		\centering
		\includegraphics[width=1\linewidth]{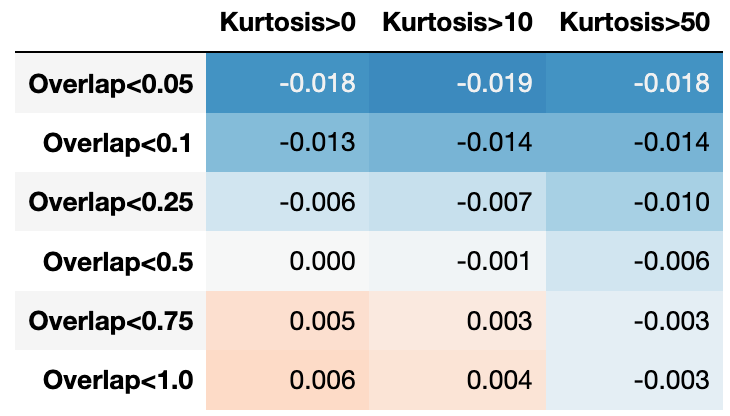}
		\vspace{-.5em}
		\caption{WMH estimation error minus JL estimation error.}
	\end{subfigure}
 \hfill
	\begin{subfigure}{0.48\columnwidth}
		\centering
		\includegraphics[width=1\linewidth]{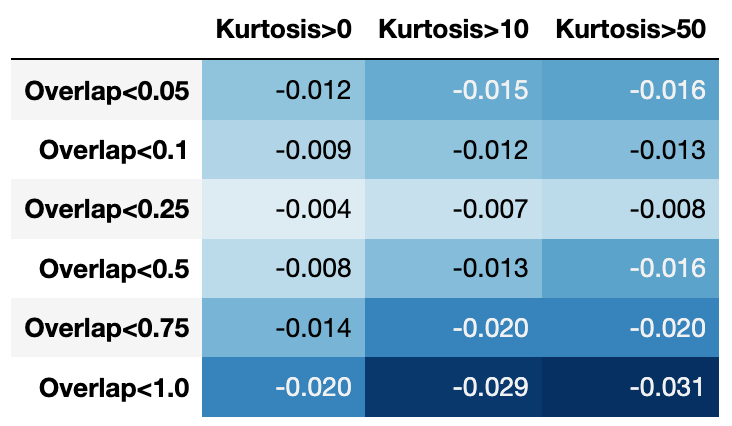}
		\vspace{-.5em}
		\caption{WMH estimation error minus MH estimation error.}
	\end{subfigure}%
	
	\vspace{-.25em}
	\caption{Inner product estimation (World Bank data). Different shades of blue highlight combinations for which WMH outperforms the other methods.}
	\label{fig:real}
\end{figure}

\begin{figure}[t]
	\centering
	
	\begin{subfigure}{.8\columnwidth} 
		\centering
		\includegraphics[width=1.0\linewidth]{figures/legend_ip.png}
	\end{subfigure}

	\begin{subfigure}{.49\columnwidth} 
		\centering
		\includegraphics[width=1.0\linewidth]{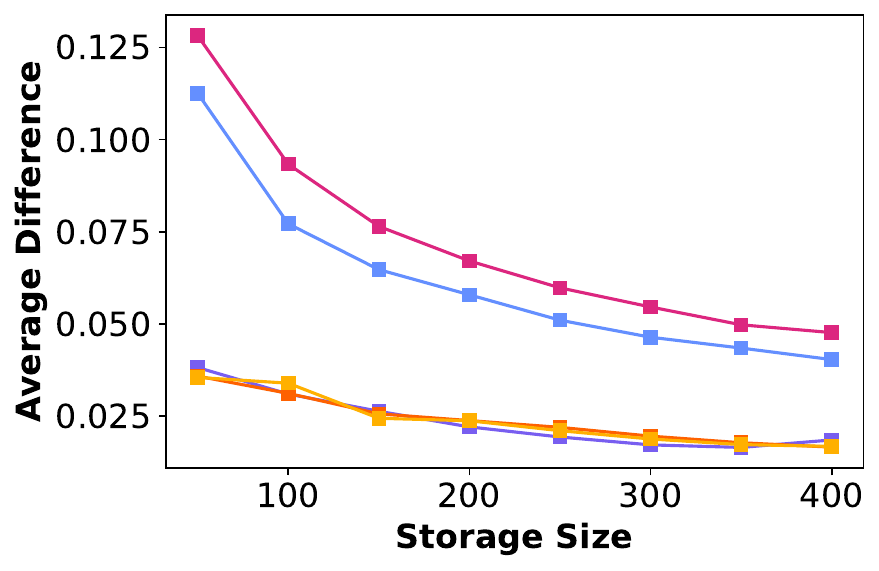}
		\vspace{-1.8em}
		\caption{All Documents}
	\end{subfigure}
	\begin{subfigure}{.49\columnwidth}
		\centering
		\includegraphics[width=1.0\linewidth]{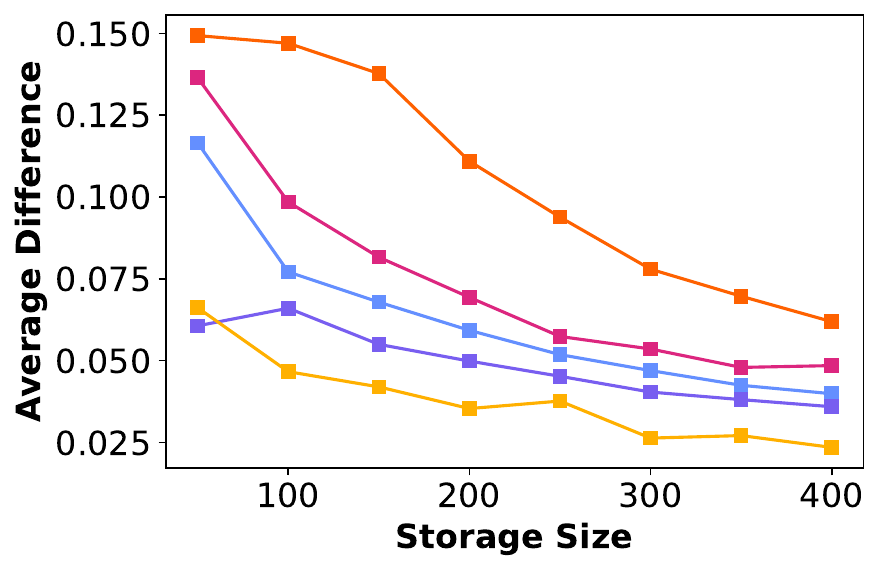}
		\vspace{-1.8em}
		\caption{Documents > 700 words}
	\end{subfigure}%
	\vspace{-.5em}
	\caption{Text similarity estimation (20 Newsgroups dataset). Note that in the left plot, the lines for MH, WMH, and KMV all lie essentially on top of one another.}
	\vspace{-.25em}
	\label{fig:tfidf}
\end{figure}

\myparagraph{
Document Similarity Estimation}
We also evaluated the performance of WMH sketches for text similarity estimation using the \textit{20 newsgroups} dataset~\cite{20newsgroups}.
We represent each document as a vector in which each entry represents a term or a combination of 2 terms (bigrams), and is associated with a value that encodes term/bigram importance using TF-IDF weights~\cite{SaltonWongYang:1975}.
This setting is well-known for generating sparse vectors of very high dimension.
As a similarity measure, we use the \textit{cosine}, which is equal to an inner product when the vectors have are normalized.
We sampled 700 documents 
and estimated the cosine similarity for 
over 200,000 pairs of documents. 
%
The results in \cref{fig:tfidf} show that, similar to previous experiments, \emph{in the worst case, the accuracy of WMH is comparable to the other methods,  but it can sometimes be better by a large margin.} In this case, it performs better for documents containing more than 700 words.
%
Note that linear projection sketches have poor performance for small sketches even when the documents are small, whereas our sampling-based methods are able to obtain significantly better accuracy for the same storage budget.
Finally,  also note that the Unweighted MinHash (MH) performs poorly for long documents whereas the weighted version still performs well.

\myparagraph{Acknowledgements}
This work was supported by the DARPA D3M program and NSF awards ISS-2106888 and CCF-2046235.  Aline Bessa was supported by a 2021 CRA/CCC CIFellows Award. Cameron Musco was also supported by a Google Research Scholar Award.
Any opinions, findings, conclusions or recommendations expressed in this material are those of the authors and do not necessarily reflect the views of NSF, DARPA, or other funding organizations.

\bibliographystyle{apalike}
\bibliography{paper}

\appendix
\section{Additional Proofs}
\label{app:add_proof}

\subsection{Unweighted MinHash Analysis}
\label{app:unweighted_proof}
\setcounter{theorem}{3}
In this section, we give a full proof of \Cref{thm:bounded}.
\vspace{-.5em}
\begin{proof}[Proof of \Cref{thm:bounded}]
	Let $\overline{\mathcal{F}}(H_{\bv{a}}, H_{\bv{b}})$ denote the estimator from  \Cref{alg:minhash_est}. Ultimately we will set $\mathcal{F}$ in \Cref{thm:bounded} to be $\overline{\mathcal{F}}$, but repeated $O(\log(1/\delta)$ times to obtain failure probability $1-\delta$. 
	
	We focus on showing first that $\overline{\mathcal{F}}(H_{\bv{a}}, H_{\bv{b}})$ achieves error $\epsilon \cdot c^2 \cdot \sqrt{\max(|\mathcal{A}|, |\mathcal{B}|)\cdot |\mathcal{A}\cap \mathcal{B}|}$ with probability $\geq 2/3$. 
	To prove this, let $\mathcal{F}^*(H_{\bv{a}}, H_{\bv{b}})$ be an alternative idealized estimator where we replace $\tilde U$ in line 1 of \Cref{alg:minhash_est} with the true union size $U = |\mathcal{A}\cup \mathcal{B}|$:
	\begin{align*}
		\mathcal{F}^*(H_{\bv{a}}, H_{\bv{b}}) =  \frac{U}{m} \sum_{i=1}^m \mathbbm{1}\left[ H_{\bv{a}}^{hash}[i] =H_{\bv{b}}^{hash}[i]\right] \cdot H_{\bv{a}}^{val}[i]\cdot H_{\bv{b}}^{val}[i].
	\end{align*}
	We will first analyze $\mathcal{F}^*$, before showing that $\overline{\mathcal{F}}$ obtains essentially as good of an estimate.
	As established in \Cref{sec:unweighted-minhash}, using the properties of \Cref{fact:mhash_sketch}, we have that
	\begin{align*}
		&\E\left[\mathcal{F}^*(H_{\bv{a}}, H_{\bv{b}})\right] = U\cdot \frac{1}{|\mathcal{A}\cup \mathcal{B}|}\cdot \langle \bv{a}, \bv{b}\rangle = \langle \bv{a}, \bv{b}\rangle.
	\end{align*}
	So we turn to bounding the variance of the estimator. Define the random variable  $Z_i = \mathbbm{1}\left[H_{\bv a}^{hash}[i] = H_{\bv b}^{hash}[i]\right]\cdot H_{\bv a}^{val}[i] \cdot H_{\bv b}^{val}[i]$ and note that $\mathcal{F}^*(H_{\bv{a}}, H_{\bv{b}}) = \frac{U}{m}\sum_{i=1}^m Z_i$.
	From  \Cref{fact:mhash_sketch} we have:
	\begin{align*}
		Z_i  =  \begin{cases}
			0 &\text{with probability } 1 - \frac{|\mathcal{A}\cap \mathcal{B}|}{|\mathcal{A}\cup \mathcal{B}|}\\
			\bv{a}[j] \bv{b}[j] &\text{with probability } \frac{1}{|\mathcal{A}\cup \mathcal{B}|} \text{ for all } j \in \mathcal{A}\cap \mathcal{B}.
		\end{cases}
	\end{align*}
	Since each $Z_i$ is independent, we can bound:
	\begin{align*}
		\Var\left[\mathcal{F}^*(H_{\bv{a}}, H_{\bv{b}}) \right] = \frac{U^2}{m^2} \sum_{i=1}^m\Var\left[Z_i \right].
	\end{align*}
	Using our assumption that $\bv{a}[k],\bv{b}[k]\leq c$ for all $k$, we have
	\begin{align*}
		\Var\left[Z_i\right] \leq \E\left[Z_i^2\right] = \sum_{j\in \mathcal{A}\cap\mathcal{B}} \frac{1}{|\mathcal{A}\cup \mathcal{B}|}  \cdot \bv{a}[j]^2 \bv{b}[j]^2 \leq c^4 \cdot \frac{|\mathcal{A}\cap \mathcal{B}|}{|\mathcal{A}\cup \mathcal{B}|},
	\end{align*}
 for all $Z_i$. 
	So we conclude that 
	$
		\Var\left[\mathcal{F}^*(H_{\bv{a}}, H_{\bv{b}}) \right] \leq  \frac{1}{m} \cdot c^4 \cdot |\mathcal{A}\cap \mathcal{B}||\mathcal{A}\cup \mathcal{B}|.
	$
	We then plug our expectation and variance bounds into Chebyhev's inequality. If $m = O(1/\epsilon^2)$, we conclude that with probability $\geq 5/6$,
	\begin{align}
		\label{eq:union1_unw}
		\left |\mathcal{F}^*(H_{\bv{a}}, H_{\bv{b}}) - \langle \bv{a},\bv{b}\rangle\right | \leq \epsilon \cdot c^2 \sqrt{|\mathcal{A}\cap \mathcal{B}||\mathcal{A}\cup \mathcal{B}|}.
	\end{align}
	
	The proof is almost complete; we just need to extend this bound to the non-idealized estimator $\overline{\mathcal{F}} = \frac{\tilde{U}}{U} \cdot \mathcal{F}^*$. We do so by observing that $\tilde{U}$ is a good approximation to $U$. Specifically, by \Cref{lem:distinct_elem} applied with $\delta = 1/6$, we have that, when $m = O(1/\epsilon^2)$, $(1-\epsilon) U \leq \tilde{U} \leq (1+\epsilon) U$,
	with probability  $\ge 5/6$.
	It follows that
	\begin{align}
		\label{eq:union2_unw}
		(1-\epsilon) \mathcal{F}^*(H_{\bv a},H_{\bv b})\leq{\overline{\mathcal{F}}(H_{\bv a},H_{\bv b})} \leq (1+\epsilon)\mathcal{F}^*(H_{\bv a},H_{\bv b}).
	\end{align}
	By a union bound, with probability at least $2/3$, both \eqref{eq:union1_unw} and \eqref{eq:union2_unw} hold simultaneously. Finally, by triangle inequality and the fact that $\langle a,b\rangle \leq c^2 |\mathcal{A}\cap \mathcal{B}| \leq c^2 \sqrt{|\mathcal{A}\cap \mathcal{B}||\mathcal{A}\cup \mathcal{B}|}$ it follows that:
	\begin{align*}
		|\overline{\mathcal{F}}(H_{\bv a},H_{\bv b})- \langle \bv{a},\bv{b}\rangle | \leq 3\epsilon\cdot c^2\cdot  \sqrt{|\mathcal{A}\cap \mathcal{B}||\mathcal{A}\cup \mathcal{B}|}.
	\end{align*}
	Noting that	$|\mathcal{A}\cap \mathcal{B}||\mathcal{A}\cup \mathcal{B}| \leq 2 \max(|\mathcal{A}|, |\mathcal{B}|)\cdot |\mathcal{A}\cap \mathcal{B}|$ and adjusting $\epsilon$ by a constant factor, we thus have that when $m = O(1/\epsilon^2)$, $\overline{\mathcal{F}}(H_{\bv a},H_{\bv b})$ satisfies the guarantee of \Cref{thm:bounded} with probability at least $ 2/3$. 
	To boost success probability to $1-\delta$, we can use the exact same median-trick used in the proof of \Cref{thm:main}: instead of computing a single pair of sketches $H_\bv{a}, H_{\bv b}$ for inputs $\bv{a},\bv{b}$, we concatenate $O(\log(1/\delta))$ sketches, each constructed using an independent random seed. If we apply $\overline{\mathcal{F}}$ to each pair of independent sketches and return the median estimate for $\langle \bv a,\bv b\rangle$, with probability at least $1-\delta$, it will satisfy our desired guarantee. 
\end{proof}

\subsection{Weighted MinHash Analysis}
\label{app:weighted_proof}
In this section we complete the analysis of \Cref{alg:weight_est} introduced in Section \ref{sec:main_result}, which yields our main result, \Cref{thm:main}. We start with a formal proof of \Cref{fact:mhash_sketchWeighted}, which is the weighted analog of \Cref{fact:mhash_sketch}.
\vspace{-.25em}
\begin{proof}[Proof of \Cref{fact:mhash_sketchWeighted}]
	Let $\bar{\mathcal A} = \{i: \bv{\bar a}[i] \neq 0\}$  and $\bar{\mathcal B} = \{i: \bv{\bar b}[i] \neq 0\}$. Since $\bv{\bar a},\bv{\bar b}$ are each comprised of $n$ blocks of $L$ elements, with the first $\bv{\tilde a}[i]^2 \cdot L$ entries and $\bv{\tilde b}[i]^2 \cdot L$ entries in the $i^\text{th}$ block set to be nonzero, we have the following equalities:
	\begin{align}
		| \bar{\mathcal{A}} \cap \bar{\mathcal{B}}| = L \cdot \sum_{j =1}^n \min(\tilde{\bv a}[j]^2,\tilde{\bv b}[j]^2)\\
		\label{eq:weightedUnion}
		| \bar{\mathcal{A}} \cup \bar{\mathcal{B}}| = L \cdot \sum_{j =1}^n \max(\tilde{\bv a}[j]^2,\tilde{\bv b}[j]^2).
	\end{align}
	Since $\bv{W}_{\bv a}^{hash}[i]$ and $\bv{W}_{\bv b}^{hash}[i]$ are constructed exactly as unweighted MinHash sketches of $\bv{\bar a}, \bv{\bar b}$, by claim (1) of \Cref{fact:mhash_sketch}, $\bv{W}_{\bv a}^{hash}[i] = \bv{W}_{\bv b}^{hash}[i]$ with probability $\frac{| \bar{\mathcal{A}} \cap \bar{\mathcal{B}}|}{| \bar{\mathcal{A}} \cup \bar{\mathcal{B}}|} = \bar J$. This gives claim (1).
	
	To prove claim (2) we note that it is equivalent to claiming that, unconditional on whether or not $\bv{W}_{\bv a}^{hash}[i] = \bv{W}_{\bv b}^{hash}[i]$, $W_{\bv a}^{val}= \bv{\tilde a}[j]$ and $W_{\bv b}^{val}= \bv{\tilde b}[j]$ for some shared $j \in \mathcal{A} \cap \mathcal{B}$ with probability $\frac{\min(\tilde{\bv a}[j]^2,\tilde{\bv b}[j]^2)}{\sum_{i=1}^n \max(\tilde{\bv a}[j]^2,\tilde{\bv b}[j]^2)}$. To prove this statement, 
	we use that, by \Cref{fact:mhash_sketch}, for any $\ell \in \mathcal{\bar A} \cap \mathcal{\bar B}$, $W_{\bv a}^{hash}[i] = W_{\bv{b}}^{hash}[i] = h^i(\ell)$, $\bv{W}_{\bv a}^{val}[i] = \bv{\bar a}[\ell]$, and $\bv{W}_{\bv b}^{val}[i] = \bv{\bar b}[\ell]$ with probability $\frac{1}{| \bar{\mathcal{A}} \cup \bar{\mathcal{B}}|} =  \frac{1}{L\sum_{k=1}^n \max(\tilde{\bv a}[k]^2,\tilde{\bv b}[k]^2)}.$ 
	Now, by construction (line 3 of \cref{alg:weighted_sketch}), $\bv{\bar a}[\ell] = \bv{\tilde a}[j]$ and $\bv{\bar b}[\ell] = \bv{\tilde b}[j]$ whenever $\ell$ lies in the $j^\text{th}$ length $L$ block of entries in $\bv{\bar a}$. For a given $j$, the number of values of $\ell$ for which $\bv{\bar a}[\ell] = \bv{\tilde a}[j]$, $\bv{\bar b}[\ell] = \bv{\tilde b}[j]$ is exactly $L \cdot \min(\bv{\tilde a}[j]^2, \bv{\tilde b}[j]^2)$. Thus, summing over these entries, $W_{\bv a}^{hash}[i] = W_{\bv{b}}^{hash}[i]$, $\bv{W}_{\bv a}^{val}[i] = \bv{\tilde a}[j]$, and $\bv{W}_{\bv b}^{val}[i] = \bv{\tilde b}[j]$ with probability $\frac{\min(\bv{\tilde a}[j]^2, \bv{\tilde b}[j]^2)}{\sum_{k=1}^n \max(\tilde{\bv a}[k]^2,\tilde{\bv b}[k]^2)}$. 
\end{proof}

\myparagraph{Analysis for Discrete Vectors}
Next, as a step towards proving \cref{thm:main}, we prove a restricted intermediate result, \Cref{lem:wminhash}, that only applies to vectors whose entries, after scaling to be unit norm, are already integer multiplies of $1/L$ for a fixed discretization parameter $L$. When this is the case, the \textsc{Round} procedure in \Cref{alg:weighted_sketch} is no-op: it simply returns $\bv{a}/\|\bv{a}\|$ unmodified. Making this assumption simplifies our analysis. Later we introduce a rounding error analysis to obtain a result for arbitrary vectors.
\begin{lemma}
	\label{lem:wminhash}
	Consider any integer discretization parameter $L$, accuracy parameter $\epsilon \in (0,1)$, and $\bv a, \bv b \in \R^n$ such that for all $i$, $\frac{\bv{ a}[i]^2}{\norm{\bv a}^2}$ and $\frac{\bv{ b}[i]^2}{\norm{\bv b}^2}$ are integer multiples of $1/L$. When run with sample size $m = O\left({1}/{\epsilon^2}\right)$ and discretization parameter $L$, \Cref{alg:weighted_sketch} returns sketches $W_{\bv a}$ and $W_{\bv b}$ such that,  letting $\mathcal{F}$ denote the estimation procedure of \Cref{alg:weight_est}, with probability at least $2/3$, 
	\begin{align*}
		|\mathcal{F}(W_{\bv a},W_{\bv b})- \langle \bv{a},\bv{b} \rangle| \leq \epsilon \max\left(\|\bv a_{\mathcal I}\|\|\bv b\|, \|\bv a\|\|\bv b_{\mathcal I}\| \right).
	\end{align*}
	Here $\mathcal{I} = \{i: \bv{a}[i] \neq 0\text{ and } \bv{b}[i] \neq 0\}$ is the intersection of $\bv a$'s  and $\bv{b}$'s supports and $\bv{a}_{\mathcal{I}}, \bv{b}_{\mathcal{I}}$ denote $\bv{a}$ and $\bv{b}$ restricted to indices in $\mathcal{I}$.
\end{lemma}
Note that \Cref{lem:wminhash} is also weaker than \Cref{thm:main} in that it only gives an accurate solution with \emph{constant probability}, $2/3$, instead of $1-\delta$ probability for any chosen $\delta$. This is again to simplify the analysis and later we show how the standard ``median-trick'' can be used to improve the success probability to $1-\delta$  \cite{CormodeGarofalakisHaas:2011,LarsenPaghTetek:2021}.

\begin{proof}
	As stated, since $\bv{ a}/\norm{\bv a}$ and $\bv{ b}/\norm{\bv b}$ have squared entries that are integer multiples of $1/L$ by assumption, in line 2 of \cref{alg:weighted_sketch}, $\textsc{Round}(\bv{a}/\norm{\bv{a}},L)$ simply sets $\bv{\tilde a} =  \bv{a}/\norm{\bv{a}}$. Analogously it sets $\bv{\tilde b} = \bv{b}/\norm{\bv{b}}$. Let $\mathcal{A} = \{i: \bv a[i] \neq 0\}$ and $\mathcal{B} = \{i: \bv b[i] \neq 0\}$ denote the supports of $\bv a$ and $\bv b$ respectively. We have $\mathcal I = \mathcal A \cap \mathcal B$.
	
	\medskip
	
	\noindent{\textbf{Reduction to Unit Vectors.}}  We first note that, to prove  the theorem, it suffices to only consider the inner product between the unit vectors $\bv{\tilde a}$ and $\bv{\tilde b}$. Specifically, we will show that:
	\begin{align}
		\label{eq:what_to_prove}
		&\left |\frac{\mathcal{F}(W_{\bv a},W_{\bv b})}{\|{\bv a}\|\|\bv b\|}- \langle \bv{\tilde{a}},\bv{\tilde{b}} \rangle\right | \\
		&\hspace{2em} \leq  \epsilon \sqrt{\sum_{i\in \mathcal{A}\cap \mathcal{B}}\max(\bv{\tilde{a}}[i]^2,\bv{\tilde{b}}[i]^2) \sum_{i=1}^n  \max(\bv{\tilde{a}}[i]^2,\bv{\tilde{b}}[i]^2) }\nonumber.
	\end{align}
	Using that $\|\bv{\tilde{a}}\|^2 + \|\bv{\tilde{b}}\|^2 = 2$ since $\bv{\tilde a}, \bv{\tilde b}$ are unit vectors, we have:
	\begin{align*}
		&\sqrt{\sum_{i\in \mathcal{A}\cap \mathcal{B}}\max(\bv{\tilde{a}}[i]^2,\bv{\tilde{b}}[i]^2) \sum_{i=1}^n \max(\bv{\tilde{a}}[i]^2,\bv{\tilde{b}}[i]^2) }\\ 
		&\hspace{2em}\leq \sqrt{\left(\|\bv{\tilde{a}}_{\mathcal{I}}\|^2 + \|\bv{\tilde{b}}_{\mathcal{I}}\|^2\right) \left(\|\bv{\tilde{a}}\|^2 + \|\bv{\tilde{b}}\|^2\right)} \\&\hspace{2em}= \sqrt{2\left(\|\bv{\tilde{a}}_{\mathcal{I}}\|^2 + \|\bv{\tilde{b}}_{\mathcal{I}}\|^2\right)} = \sqrt{2\left(\frac{\|{\bv a}_{\mathcal{I}}\|^2}{\|\bv a\|^2} + \frac{\|{\bv b}_{\mathcal{I}}\|^2}{\|\bv b\|^2} \right)}.
	\end{align*}
	Thus,
	multiplying \eqref{eq:what_to_prove} on both sides by $\|\bv a\|\|\bv b\|$ we  have: 
	\begin{align*}
		\left|\mathcal{F}(W_{\bv a},W_{\bv b})- \langle \bv{{a}},\bv{{b}} \rangle\right|  &\le \epsilon\sqrt{2} \|\bv a\|\|\bv b\| \cdot \sqrt{\frac{\|{\bv a}_{\mathcal{I}}\|^2}{\|\bv a\|^2} + \frac{\|{\bv b}_{\mathcal{I}}\|^2}{\|\bv b\|^2}}\\
		&=  \epsilon\sqrt{2}\sqrt{\|{\bv a}_{\mathcal I}\|^2\|\bv b\|^2 + \|{\bv b}_{\mathcal{I}}\|^2\|\bv a\|^2 }\\
		& \leq {2}\epsilon \cdot \max\left(\|{\bv a}_{\mathcal I}\|\|\bv b\|, \|{\bv b}_{\mathcal{I}}\|\|\bv a\| \right).
	\end{align*}
	The last inequality follows from the fact that the sum is at most two times the max.
	Adjusting $\epsilon$ by a constant gives the desired bound of \cref{lem:wminhash}.
	Thus, we turn our attention to proving  \eqref{eq:what_to_prove}. 
	
	\medskip 
	
	\noindent \textbf{Analysis for Unit Vectors.} 
	We start by analyzing an idealized version of the estimator computed by \cref{alg:weight_est}, where $M$ is replaced by the \emph{exact} weighted union size $M = \sum_{i=1}^n \max(\tilde{\bv a}[i]^2,\tilde{\bv b}[i]^2).$ Specifically, define:
	\begin{align}
		\label{eq:ideal_estimate}
		\mathcal{F}^* = \frac{M}{m} \sum_{i=1}^m \mathbbm{1}\left[W_{\bv a}^{hash}[i] = W_{\bv b}^{hash}[i]\right]\cdot \frac{W_{\bv a}^{val}[i] \cdot W_{\bv b}^{val}[i]}{q_i},
	\end{align}
	where $q_i = \min\left(W_{\bv{a}}^{val}[i]^2,W_{\bv{b}}^{val}[i]^2\right)$ as  in line 1 of \cref{alg:weight_est}.
	
	We first show that $\E[\mathcal{F}^*] = \langle \tilde{\bv a}, \tilde{\bv b}\rangle$ and then bound $\mathcal{F}^*$'s variance. 
	For each $i \in \{1,\ldots,m\}$ define the random variable $Z_i$ as
	\begin{align*}
		Z_i =\mathbbm{1}\left[W_{\bv a}^{hash}[i] = W_{\bv b}^{hash}[i]\right]\cdot \frac{W_{\bv a}^{val}[i] \cdot W_{\bv b}^{val}[i]}{q_i}.
	\end{align*}
	Recalling that $\bar{J} = \frac{\sum_{j=1}^n \min(\bv{\tilde{a}}[j]^2,\tilde{\bv b}[j]^2)}{\sum_{j=1}^n \max(\bv{\tilde{a}}[j]^2,\tilde{\bv b}[j]^2)}$ is the weighted Jaccard similarity between $\bv{\tilde{a}}$ and $\bv{\tilde{b}}$, applying \Cref{fact:mhash_sketchWeighted} we have:
	\begin{align*}
		Z_i  =  \begin{cases}
			0 &\text{with probability } 1 - \bar{J}\\
			\frac{\tilde{\bv a}[j] \tilde{\bv b}[j]}{\min(\tilde{\bv a}[j]^2,\tilde{\bv b}[j]^2)} &\text{with probability } \frac{\min(\tilde{\bv a}[j]^2,\tilde{\bv b}[j]^2)}{\sum_{k=1}^n \max(\tilde{\bv a}[k]^2,\tilde{\bv b}[k]^2)}\\ &\text{for all } j \in \mathcal{A} \cap \mathcal{B}.
		\end{cases}
	\end{align*}
	Thus, $\E[Z_i] =\frac{\langle \tilde{\bv a},\tilde{\bv b}\rangle}{\sum_{k = 1}^n \max(\tilde{\bv a}[k]^2,\tilde{\bv b}[k]^2)} = \frac{\langle \tilde{\bv a},\tilde{\bv b}\rangle}{M}$. Since $\mathcal{F}^* =\frac{M}{m}\sum_{i=1}^m Z_i$, it follows from linearity of expectation that:
	\begin{align}
		\label{eq:weighted_expectation}
		\E[\mathcal{F}^*] =  \frac{M}{m}  \sum_{i=1}^m \E\left[Z_i\right] = \langle \bv{\tilde a}, \bv{\tilde b}\rangle.
	\end{align}
	We next bound the variance of $\mathcal{F}^*$. For each $Z_i$ we have that: 
	\begin{align*}
		\Var&[Z_i] \leq \sum_{j \in \mathcal A \cap \mathcal B} \frac{\min(\tilde{\bv a}[j]^2,\tilde{\bv b}[j]^2)}{\sum_{k=1}^n \max(\tilde{\bv a}[k]^2,\tilde{\bv b}[k]^2)} \cdot \frac{\tilde{\bv a}[j]^2 \tilde{\bv b}[j]^2}{\min(\tilde{\bv a}[j]^2,\tilde{\bv b}[j]^2)^2}\\ &= \sum_{j \in \mathcal A \cap \mathcal B} \frac{\max(\tilde{\bv a}[j]^2,\tilde{\bv b}[j]^2)}{\sum_{k=1}^n \max(\tilde{\bv a}[k]^2,\tilde{\bv b}[k]^2)} \\
		& =\frac{\sum_{j\in \mathcal{A}\cap \mathcal{B}} \max(\tilde{\bv a}[j]^2,\tilde{\bv b}[j]^2)}{\sum_{k=1}^n \max(\tilde{\bv a}[k]^2,\tilde{\bv b}[k]^2)} = \frac{\sum_{j\in \mathcal{A}\cap \mathcal{B}} \max(\tilde{\bv a}[j]^2,\tilde{\bv b}[j]^2)}{M}
	\end{align*} 
	Since each $Z_i$ is independent, it follows that:
	\begin{align}
		\label{eq:weighted_variance}
		\Var[&\mathcal{F}^*] = \frac{M^2}{m^2}\sum_{i=1}^m \Var\left[Z_i\right]\nonumber\\
		&\le \frac{1}{m}\sum_{j\in \mathcal{A}\cap \mathcal{B}} \max(\tilde{\bv a}[j]^2,\tilde{\bv b}[j]^2) \cdot  \sum_{j=1}^n \max(\tilde{\bv a}[j]^2,\tilde{\bv b}[j]^2) . 
	\end{align}
	Combining \eqref{eq:weighted_expectation} and \eqref{eq:weighted_variance} with Chebyshev's inequality, we can claim that when $m = O(1/\epsilon^2)$, with probability at least $5/6$:
	\begin{align}\label{eq:union1}
		\left|\mathcal{F}^*- \langle \tilde{\bv a},\tilde{\bv b} \rangle\right| \leq \epsilon \sqrt{\sum_{j\in \mathcal{A}\cap \mathcal{B}} \max(\tilde{\bv a}[j]^2,\tilde{\bv b}[j]^2)  \sum_{j=1}^n \max(\tilde{\bv a}[j]^2,\tilde{\bv b}[j]^2) }.
	\end{align}
	
	We want to extend this bound from the idealized estimator $\mathcal{F}^*$ to our true estimator $\mathcal{F}$, 
	which equals $\frac{\tilde{M}}{M}\cdot \mathcal{F}^*$. 
	To do so, we use that  $\tilde{M}$ is a good approximation to $M$. 
	As discussed in \Cref{sec:main_result}, this is because $\tilde{M}$ exactly equals $\frac{1}{L}$ times a distinct elements estimator applied to the support sets $\bar{\mathcal{A}}$ and $\bar{\mathcal{B}}$ 
	of the extended vectors $\bv{\bar a}, \bv{\bar b}$. From \eqref{eq:weightedUnion} and  \Cref{lem:distinct_elem}, we have that for $m = O(1/\epsilon^2)$,
	\begin{align*}
		(1-\epsilon) M \leq \tilde{M} \leq (1+\epsilon) M,
	\end{align*}
	with probability at least $5/6$.
	It follows that:
	\begin{align}\label{eq:union2}
		(1-\epsilon) \mathcal{F}^*\leq\frac{\mathcal{F}(W_{\bv a},W_{\bv b})}{\|{\bv a}\|\|\bv b\|}  \leq (1+\epsilon)\mathcal{F}^*.
	\end{align}
	By a union bound, with probability at least $2/3$, both \eqref{eq:union1} and \eqref{eq:union2} hold simultaneously. Finally, by Cauchy-Schwarz inequality,
	\begin{align*}\langle \tilde{\bv a},\tilde{\bv b}\rangle \leq \sqrt{\sum_{j\in \mathcal{A}\cap \mathcal{B}} \max(\tilde{\bv a}[j]^2,\tilde{\bv b}[j]^2)  \sum_{j=1}^n \max(\tilde{\bv a}[j]^2,\tilde{\bv b}[j]^2) }.
	\end{align*}
	Combining \eqref{eq:union1} and \eqref{eq:union2} with triangle inequality, it follows that
	\begin{align*}
		&\left|\frac{\mathcal{F}(W_{\bv a},W_{\bv b})}{\|{\bv a}\|\|\bv b\|} - \langle \tilde{\bv a},\tilde{\bv b}\rangle\right| \\&\hspace{2em}\leq 3\epsilon \sqrt{\sum_{j\in \mathcal{A}\cap \mathcal{B}} \max(\tilde{\bv a}[j]^2,\tilde{\bv b}[j]^2)  \sum_{j=1}^n \max(\tilde{\bv a}[j]^2,\tilde{\bv b}[j]^2) }.
	\end{align*}
	Adjusting  $\epsilon$ by a $1/3$ factor proves \Cref{lem:wminhash}.
\end{proof}

\myparagraph{Rounding for Continuous Vectors}
With \Cref{lem:wminhash} in place, we complete our proof of \Cref{thm:main} by analyzing the impact of the rounding step in \Cref{alg:weight_est}. In \Cref{lem:round}, we show that if $L$ is set on the order of $n^6/\epsilon^2$, then we can bound the impact of this step on the accuracy of our inner product estimate. 
Formally, we have:
\begin{lemma}[Rounding]\label{lem:round}
	Consider any $\bv{a}, \bv{b} \in \R^n$ and discretization parameter $L$. Let $\bv{\tilde a} = \textsc{Round}(\bv{a}/\norm{\bv a},L)$ and $\bv{\tilde b} = \textsc{Round}(\bv{b}/\norm{\bv b},L)$, as in line 2 of \cref{alg:weighted_sketch}. Let $\bv{a}' = \norm{\bv a} \cdot \bv{\tilde a}$ and $\bv{b}' = \norm{\bv b} \cdot \bv{\tilde b}$, and let $B$ denote $B =\max\left(\|\bv a_{\mathcal I}\|\|\bv b\|, \|\bv a\|\|\bv b_{\mathcal I}\| \right).$
	\begin{enumerate}
		\item $\bv{a}',\bv{b}'$ satisfy the assumption of \Cref{lem:wminhash}, that for all $i$, $\frac{\bv{ a}'[i]^2}{\norm{\bv a'}^2}$ and $\frac{\bv{ b}'[i]^2}{\norm{\bv b'}^2}$ are integer multiples of $1/L$.
		\item For any discretization parameter $L$, sketch size $m$, and random seed $s$, \Cref{alg:weighted_sketch} yields identical outputs on $\bv{a},\bv{b}$ and $\bv{a}', \bv{b}'$. I.e., $W_{\bv a} = W_{\bv a'}$ and $W_{\bv b} = W_{\bv b'}$.
		\item For $L \geq 9n^6/\epsilon^2$,  $\left |\langle \bv a,\bv b\rangle - \langle \bv a',\bv b'\rangle \right |\le \epsilon B.$
		\item For $L \geq n^3$,  $\max\left(\|\bv a'_{\mathcal I}\|\|\bv b'\|, \|\bv a'\|\|\bv b'_{\mathcal I}\| \right) \le 2 B.$
	\end{enumerate}
\end{lemma}
\begin{proof}
	We prove the four claims of the lemma in order. For the first two, we focus on $\bv{a}$ and $\bv{a}'$. Identical claims hold for $\bv{b}$ and $\bv{b}'$.
		\smallskip
		
		\noindent\textbf{Claim 1:  $\frac{\bv{ a}'[i]^2}{\norm{\bv a'}^2}$ is an integer multiple of $1/L$ for all $i$}.  First observe that $\bv{\tilde a} = \textsc{Round}(\bv{a}/\norm{\bv a},L)$ is a unit vector. This is ensured by line 3 of \cref{alg:round}. Thus, $\norm{\bv{a}'} = \norm{\bv a} \cdot \norm{\bv{\tilde a}} = \norm{\bv a}$ and $\frac{\bv{a}'[i]^2}{\norm{\bv{a}'}^2} = \frac{\bv{a}'[i]^2}{\norm{\bv a}^2} = \bv{\tilde a}[i]^2$. 
		So to prove the claim, it suffices to show that $\bv{\tilde a}[i]^2$ is an integer multiple of $1/L$ for all $i$. This is guaranteed by \cref{alg:round}. After line 1, we can see that $\bv{\tilde z}[i]^2$ is an integer multiple of $1/L$ for all $i$. Since $L$ is an integer, $1$ is also trivially an integer multiple of $1/L$. So $\delta = 1-\norm{\bv{\tilde z}}^2$ as set in line 2 is an integer multiple of $1/L$. Finally, this ensures that $\bv{\tilde z}[i^*]^2 = \bv{\tilde z}[i^*]^2 +\delta$ as set in line 3 is an integer multiple of $1/L$, completing the claim.
		
		\medskip
		
		\noindent\textbf{Claim 2:} $W_{\bv a}=W_{\bv a'}$. As shown above, $\norm{\bv a'} = \norm{\bv a}$. So to prove the claim, it  suffices to show that $\textsc{Round}\left (\frac{\bv{a}}{\|\bv{a}\|}, L\right) = \textsc{Round} \left (\frac{\bv{a}'}{\|\bv{a}'\|}, L\right )$. This ensures that \cref{alg:weighted_sketch} proceeds identically on inputs $\bv a$ and $\bv a'$. By Claim (1), $\textsc{Round}(\bv{a}'/\|\bv{a}'\|, L) = \bv{a}'/\norm{\bv a'} = \bv{a}'/\norm{\bv a} = \bv{\tilde a}$. And by definition, $\bv{\tilde a} = \textsc{Round}(\bv{a}/\|\bv{a}\|, L)$. This completes the claim.
		\medskip
		
		\noindent\textbf{Claim 3:} For $L \geq 9n^6/\epsilon^2$, $\left |\langle \bv a,\bv b\rangle - \langle \bv a',\bv b'\rangle \right |\le \epsilon B.$
		%
		%
		Let $\bv{\hat a} = \bv{a}/\norm{\bv a}$ and $\bv{\hat b} = \bv{b}/\norm{\bv b}$. So $\bv{\tilde a} = \textsc{Round}(\bv{\hat a},L)$ and $\bv{\tilde b} = \textsc{Round}(\bv{\hat b},L)$.
		We will show that
		\begin{align}\label{eq:roundImpact}
			\left | \langle \bv{\hat a},\bv{\hat b}\rangle - \langle \bv{\tilde a},\bv{\tilde b} \rangle \right | \le \epsilon \cdot \sqrt{\|\bv{\hat{a}}_{\mathcal{I}}\|^2 + \|\bv{\hat{b}}_{\mathcal{I}}\|^2}.
		\end{align} 
		Multiplying each side of \eqref{eq:roundImpact} by $\norm{\bv a} \norm{\bv b}$ then gives:
		\begin{align*}
			\left | \langle \bv{ a},\bv{ b}\rangle - \langle \bv{ a}',\bv{ b}' \rangle \right | &\le \epsilon \cdot \norm{\bv a} \norm{\bv b} \sqrt{\|\bv{\hat{a}}_{\mathcal{I}}\|^2 + \|\bv{\hat{b}}_{\mathcal{I}}\|^2} \\
			&= \epsilon \cdot \norm{\bv a} \norm{\bv b} \sqrt{\left(\frac{\|\bv{{a}}_{\mathcal{I}}\|^2}{\norm{\bv a}^2} + \frac{\|\bv{{b}}_{\mathcal{I}}\|^2}{\norm{\bv b}^2}\right)}\\
			&=  \epsilon\sqrt{\left(\|{\bv a}_{\mathcal I}\|^2\|\bv b\|^2 + \|{\bv b}_{\mathcal{I}}\|^2\|\bv a\|^2 \right)}\\
			& \leq \sqrt{2}\epsilon \cdot \max\left(\|{\bv a}_{\mathcal I}\|^2\|\bv b\|^2, \|{\bv b}_{\mathcal{I}}\|^2\|\bv a\|^2 \right),
		\end{align*}
		which completes the claim after adjusting $\epsilon$ by a constant.
		
		We proceed to prove \eqref{eq:roundImpact}. 
		Observe that for any $i \notin \mathcal{I}$, we have at least one of $\bv{\hat a}[i]$ or $\bv{\hat b}[i]$ equal to $0$. In turn, at least one of $\bv{\tilde a}[i]$ or $\bv{\tilde b}[i]$ is also $0$ since in the rounding procedure of \cref{alg:round} any entry of $\bv{z}$ that is $0$ is set to $0$ in $\bv{\tilde z}$. 
		So we can conclude that $\langle \bv{\hat a},\bv{\hat b} \rangle = \langle \bv{\hat a}_{\mathcal I},\bv{\hat b}_{\mathcal I} \rangle$ and similarly, $\langle \bv{\tilde a},\bv{\tilde b} \rangle = \langle \bv{\tilde a}_{\mathcal I},\bv{\tilde b}_{\mathcal I} \rangle$. This gives that:
		\begin{align*}
			\left | \langle \bv{\hat a},\bv{\hat b}\rangle - \langle \bv{\tilde a},\bv{\tilde b} \rangle \right | = \left | \langle \bv{\hat a}_{\mathcal{I}},\bv{\hat b}_{\mathcal{I}} \rangle - \langle \bv{\tilde a}_{\mathcal I},\bv{\tilde b}_{\mathcal{I}}\rangle \right |.
		\end{align*}
		So, to prove \eqref{eq:roundImpact}, it suffices to bound the righthand side of the above equation. We  consider two cases:
		
		\smallskip
		
		\noindent\textbf{Case 1: $\max \left (\norm{\bv{\hat a}_{\mathcal I}},\norm{\bv{\hat b}_{\mathcal I}} \right ) \ge \frac{1}{\sqrt{L}}$.}
		For $i \in \mathcal{I}$, if $|\bv{\hat a}[i]| < \frac{1}{\sqrt{L}}$ and $L \ge n$, then $|\bv{\hat a}[i]| < \frac{1}{\sqrt{n}}$ and so $i \neq \argmax_{i \in 1,\ldots,n} \bv{\hat a}[i]$ since $\bv{\hat a}$ is a unit vector so has at least one entry with magnitude $\ge 1/\sqrt{n}$. Thus, $\bv{\hat a}[i]$ is rounded in line 1 of \cref{alg:round}, and not in line 3. We have $\lfloor \bv{\hat a}[i]^2 \cdot L \rfloor = 0$ and so 
		$|\tilde{\bv a}[i] - \bv{\hat a}[i] | = |\hat{\bv a}[i]| < \frac{1}{\sqrt{L}}$. Alternatively, if $|\bv{\hat a}[i]| \ge \frac{1}{\sqrt{L}}$ and $i \neq \argmax_{i \in 1,\ldots,n} \bv{\hat a}[i]$ (so $\bv{\hat a}[i]$ is rounded in line 1 but not line 3 of \cref{alg:round}) then:
		\begin{align*}
			|\tilde{\bv a}[i] - \bv{\hat a}[i] | &\le \frac{1}{\sqrt{L}} \cdot \left |\sqrt{\bv{\hat a}[i]^2 \cdot L} - \sqrt{\bv{\hat a}[i]^2 \cdot L-1} \right | \\&= \frac{1}{\sqrt{L}} \cdot \frac{1}{\sqrt{\bv{\hat a}[i]^2 \cdot L} + \sqrt{\bv{\hat a}[i]^2 \cdot L-1}} \le \frac{1}{\sqrt{L}}.
		\end{align*}
		If $i = \argmax_{i \in 1,\ldots,n} \bv{\hat a}[i]$ then $\bv{\hat a}[i]$ is rounded in line 3 and so
		\begin{align}\label{eq:deltaBound}
			|\tilde{\bv a}[i] - \bv{\hat a}[i] | \le \left | \sqrt{\bv{\hat a}[i]^2 + \delta} - |\bv{\hat a}[i]| \right | \le \frac{\delta}{2 |\bv{\hat a}[i]|} ,
		\end{align}
		where we use that $\sqrt{x}$ is concave with derivative $\frac{1}{{2 |\bv{\hat a}[i]|}}$ at $\bv{\hat a}[i]^2$. In line 2 of \cref{alg:round} we set $\delta = 1 - \norm{\bv{\tilde a}}^2$, where $\bv{\tilde a}$ is formed by rounding down entries of $\bv{\hat a}$ in line 1. Each squared entry is rounded down by at most $1/L$, so recalling that $\bv{\hat a}$ is a unit vector, $\delta \le n/L$.
		Plugging into \eqref{eq:deltaBound}, and recalling that we assume $\bv{\hat a}[i] \ge 1/\sqrt{L}$,
		\begin{align}\label{eq:deltaBound2}
			|\tilde{\bv a}[i] - \bv{\hat a}[i] | \le \frac{n/L}{2/ \sqrt{L}} \le \frac{n}{\sqrt{L}}.
		\end{align}
		Overall, we can conclude that $\norm{\bv{\tilde a}_{\mathcal{I}}-\bv{\hat a}_{\mathcal I}}_\infty \le \frac{n}{\sqrt{L}}$. 
		Similarly, we have $\norm{\bv{\tilde b}_{\mathcal I}- \bv{\hat b}_{\mathcal I}}_\infty \le \frac{n}{\sqrt{L}}$. Thus,
		\begin{align*}
			\left | \langle \bv{\tilde a},\bv{\tilde b}\rangle - \langle \bv{\hat a}, \bv{\hat b} \rangle \right | &= \left | \langle \bv{\tilde a}_{\mathcal{I}},\bv{\tilde b}_{\mathcal{I}} \rangle - \langle \bv{\hat a}_{\mathcal I},\bv{\hat b}_{\mathcal{I}}\rangle \right | \\&\le \frac{n}{\sqrt{L}} \left (\norm{\bv{\hat a}_{\mathcal I}}_1 + \norm{\bv{\hat b}_{\mathcal I}}_1 \right ) + \frac{|\mathcal{I}| \cdot n^2}{L}.
		\end{align*}
		By Cauchy-Schwarz, we have $\norm{\bv{\hat a}_{\mathcal{I}}}_1 \le \sqrt{|\mathcal I|} \cdot \norm{\bv{\hat a}_{\mathcal{I}}}$ and $\norm{\bv{\hat b}_{\mathcal{I}}}_1 \le \sqrt{|\mathcal I|} \cdot \norm{\bv{\hat b}_{\mathcal{I}}}$. Overall, this gives:
		\begin{align*}
			\left | \langle \bv{\tilde a},\bv{\tilde b}\rangle - \langle \bv{\hat a}, \bv{\hat b} \rangle \right |  &\le \frac{n\sqrt{|\mathcal I|}}{\sqrt L} \left (\norm{\bv{\hat a}_{\mathcal{I}}} + \norm{\bv{\hat b}_{\mathcal{I}}} \right )+ \frac{{|\mathcal I|\cdot n^2}}{L}\\
			&\le \frac{n^3}{\sqrt{L}} \cdot \left ( \norm{\bv{\hat a}_{\mathcal{I}}} + \norm{\bv{\hat b}_{\mathcal{I}}} + \max \left (\norm{\bv{\hat a}_{\mathcal I}},\norm{\bv{\hat b}_{\mathcal I}} \right ) \right ),
		\end{align*}
		where in the last line we use that $|\mathcal I| \le n$, along with the assumption of Case 1 that $\max \left (\norm{\bv{\hat a}_{\mathcal I}},\norm{\bv{\hat b}_{\mathcal I}} \right ) \ge \frac{1}{\sqrt{L}}$.
		Setting $L \ge \frac{9n^6}{\epsilon^2}$, we have
		\begin{align*}
			\left | \langle \bv{\tilde a},\bv{\tilde b}\rangle - \langle \bv{\hat a},\bv{\hat b} \rangle \right | &\le \epsilon \cdot \max \left (\norm{\bv{\hat a}_{\mathcal I}},\norm{\bv{\hat b}_{\mathcal I}} \right ) \le \epsilon \cdot \sqrt{\|\bv{\hat{a}}_{\mathcal{I}}\|^2 + \|\bv{\hat{b}}_{\mathcal{I}}\|^2}.
		\end{align*}
		This proves \eqref{eq:roundImpact} for Case 1.
		
		\smallskip
		
		\noindent\textbf{Case 2: $\max \left (\norm{\bv{\hat a}_{\mathcal I}},\norm{\bv{\hat b}_{\mathcal I}} \right ) < \frac{1}{\sqrt{L}}$.} In this case, for all $i \in \mathcal{I}$,  $|\hat{\bv a}[i]| < \frac{1}{\sqrt{L}}$ and $|\hat{\bv b}[i]| < \frac{1}{\sqrt{L}}$. Thus, for $L > n$, no $i \in \mathcal{I}$ satisfies $i = \argmax_{i \in 1,\ldots,n} \bv{\hat a}[i]$ or $i = \argmax_{i \in 1,\ldots,n} \bv{\hat b}[i]$. So for all $i \in \mathcal{I}$, $\bv{\hat a}[i]$ and $\bv{\hat b}[i]$ are rounded to $0$ in line 1 of \cref{alg:round}. I.e.,
		$\tilde{\bv a}_{\mathcal{I}}$ and $\tilde{\bv b}_{\mathcal{I}}$ are both all zero vectors. So, to prove \eqref{eq:roundImpact}, we must show that 
		$
			\left | \langle \bv{\hat a}_{\mathcal{I}},\bv{\hat b}_{\mathcal I}\rangle \right | \le \epsilon \cdot \sqrt{\|\bv{\hat{a}}_{\mathcal{I}}\|^2 + \|\bv{\hat{b}}_{\mathcal{I}}\|^2}.
		$
		This follows from Cauchy-Schwarz and our assumption that $\norm{\bv a_{\mathcal I}},\norm{\bv b_{\mathcal I}} < \frac{1}{\sqrt{L}}$
		\begin{align*}
			\left | \langle \bv{\hat a}_{\mathcal{I}},\bv{\hat b}_{\mathcal I}\rangle \right | \le \norm{ \bv{\hat a}_{\mathcal{I}}} \norm{ \bv{\hat b}_{\mathcal{I}}} &\le \frac{1}{\sqrt{L}}\max\left (\norm{ \bv{\hat a}_{\mathcal{I}}}, \norm{ \bv{\hat b}_{\mathcal{I}}}\right )\\
			& \le \frac{1}{\sqrt{L}} \sqrt{\norm{\bv{\hat a}_{\mathcal{I}}}^2 +\norm{ \bv{\hat b}_{\mathcal{I}}}^2}.
		\end{align*}
Setting $L \ge \frac{1}{\epsilon^2}$ gives \eqref{eq:roundImpact}, completing Claim (3) of the lemma.
		
		\medskip
		
		\noindent\textbf{Claim 4:}  For $L \geq n^3$,  $\max\left(\|\bv a'_{\mathcal I}\|\|\bv b'\|, \|\bv a'\|\|\bv b'_{\mathcal I}\| \right) \le 2 B.$  Recall that by construction $\norm{\bv{a}'} = \norm{\bv{a}}$ and $\norm{\bv{b}'} = \norm{\bv{b}}$. Thus, dividing each side of the inequality by $\norm{\bv a} \norm{\bv b}$ it suffices to show:
		\begin{align*}
			\max\left(\frac{\|\bv a'_{\mathcal I}\|}{\norm{\bv a}}, \frac{\|\bv b'_{\mathcal I}\|}{\norm{\bv b}} \right) \le 2 \max\left(\frac{\|\bv a_{\mathcal I}\|}{\|\bv a\|}, \frac{\|\bv b_{\mathcal I}\|}{\|\bv b\|} \right).
		\end{align*}
		I.e., we must show that $
		\max(\norm{\bv{\tilde a}_{\mathcal I}}, \norm{\bv{\tilde b}_{\mathcal I}} ) \le 2 \max(\norm{\bv{\hat a}_{\mathcal I}}, \norm{\bv{\hat b}_{\mathcal I}}).
		$
		It suffices to show that $\norm{\bv{\tilde a}_{\mathcal{I}}} \le 2 \norm{\bv{\hat a}_{\mathcal I}}$ and that $\norm{\bv{\tilde b}_{\mathcal{I}}} \le 2 \norm{\bv{\hat b}_{\mathcal I}}$. We focus on proving this for $\bv{a}$. The bound for $\bv{b}$ follows  the same argument.
		We consider two cases. Let $i^* = \argmax_{i \in 1,\ldots,n} |\bv{\hat a}[i]|$.
		
		\smallskip
		
		\noindent\textbf{Case 1: $i^* \notin \mathcal{I}$.} In this case, all entries in $\bv{\hat a}_{\mathcal{I}}$ are only rounded in line 1 of \cref{alg:round}. They are thus all rounded down and so $\norm{\bv{\tilde a}_{\mathcal{I}}} \le \norm{\bv{\hat a}_{\mathcal I}}$, giving the claim.
		
		\smallskip
		
		\noindent\textbf{Case 2: $i^* \in \mathcal{I}$.} In this case, since $\bv{\hat a}$ is a unit vector, we have $\norm{\bv{\hat a}_\mathcal{I}} \ge |\bv{\hat a}[i^*]| \ge 1/\sqrt{n} \ge 1/\sqrt{L}$ when $L > n$.
		Further, all entries in $\bv{\hat a}_{\mathcal I}$ are rounded down, except $\bv{\hat a}[i^*]$. But as shown via \eqref{eq:deltaBound2}, $|\bv{\tilde a}[i^*]| \le |\bv{\hat a}[i^*]| + \frac{n}{\sqrt{L}}.$ Thus, $\norm{\bv{\tilde a}_{\mathcal I}} \le \norm{\bv{\hat a}_{\mathcal I}} + \frac{n}{\sqrt{L}} \le 2 \norm{\bv{\hat a}_{\mathcal I}}$, as long as $L \ge n^3$.
		
		This completes Claim (4) and thus the lemma.
\end{proof}

\myparagraph{Putting everything together}
Finally, we prove our main result by combining \Cref{lem:round} with \Cref{lem:wminhash}.
\begin{proof}[Proof of \cref{thm:main}]
	Given any $\bv a,\bv b \in \R^n$, let
	$\bv{a}'$ and $\bv{b}'$ be defined as in \cref{lem:round}. Consider applying \cref{alg:weighted_sketch} to compute sketches $W_{\bv a}, W_{\bv b}, W_{\bv a'}, W_{\bv b'}$ of size $m = O(1/\epsilon^2)$, using discretization parameter $L =O(n^6/\epsilon^2)$.
	Using the first claim of \Cref{lem:round} , we can apply \cref{lem:wminhash} to $\bv{a}',\bv{b}'$ to show that with probability $\ge 2/3$,
	\begin{align*}
		|\mathcal{F}(W_{\bv a'},W_{\bv b'})- \langle \bv a',\bv b' \rangle| \leq \epsilon \max\left(\|\bv a'_{\mathcal I}\|\|\bv b'\|, \|\bv a'\|\|\bv b'_{\mathcal I}\| \right).
	\end{align*} 
	Combining triangle inequality with Claims (2) and (4) of \cref{lem:round}, we conclude  that with probability $\ge 2/3$,
	\begin{align*}
		|\mathcal{F}(W_{\bv a},W_{\bv b})- \langle \bv a,\bv b \rangle| &\le \left |\langle \bv a,\bv b\rangle - \langle \bv a',\bv b'\rangle \right | \\ &\hspace{2em}+ 2 \epsilon  \max\left(\|\bv a_{\mathcal I}\|\|\bv b\|, \|\bv a\|\|\bv b_{\mathcal I}\| \right).
	\end{align*}
	Finally, applying Claim (3) of \cref{lem:round} gives that
	\begin{align*}
		|\mathcal{F}(W_{\bv a},W_{\bv b})- \langle \bv a,\bv b \rangle| \le 3\epsilon \max\left(\|\bv a_{\mathcal I}\|\|\bv b\|, \|\bv a\|\|\bv b_{\mathcal I}\| \right ).
	\end{align*}
	After adjusting $\epsilon$ by a factor of $1/3$, this establishes the bound of \cref{thm:main}. The probability of success is $2/3$. Using a standard trick, we can boost the success probability by computing $t = O(\log(1/\delta))$ independent sketches of $\bv a,\bv b$ using \cref{alg:weighted_sketch} with independent random seeds \cite{CormodeGarofalakisHaas:2011}. Call these sketches $W^{(1)}_{\bv a},\ldots, W^{(t)}_{\bv a}$ and $W^{(1)}_{\bv b},\ldots, W^{(t)}_{\bv b}$. For any $i$,  with probability $\ge 2/3$, 
	\begin{align*}
		|\mathcal{F}(W^{(i)}_{\bv a},W^{(i)}_{\bv b})- \langle \bv a,\bv b \rangle| \le \epsilon\max\left(\|\bv a_{\mathcal I}\|\|\bv b\|, \|\bv a\|\|\bv b_{\mathcal I}\| \right ).
	\end{align*} 
	Via a standard Chernoff bound, with probability at least $1-\delta$, this bound holds for $> t/2$ of the independent sketches. Thus, if we take the median estimate produced by the sketches, it will satisfy the desired bound with probability $\ge 1-\delta$. Concatenating our $t$ independent sketches into a single sketch,  we can see that the total sketch size is $t \cdot m = O(\log(1/\delta)/\epsilon^2)$, giving  \cref{thm:main}.
\end{proof}

\end{document}